\newif\ifarxiv
\arxivtrue

\ifarxiv
\documentclass[11pt]{article}
\else
\documentclass[a4paper,UKenglish,cleveref, autoref, thm-restate]{lipics-v2021}
\fi

\ifarxiv
\usepackage[english]{babel}
\usepackage{times}
\usepackage[margin=1in]{geometry}
\usepackage{enumitem}
\setlist{nosep}
\else
\fi

\usepackage[utf8]{inputenc}
\usepackage{amsfonts,amsmath,amssymb, mathtools}
\usepackage{amsthm, thmtools, thm-restate}
\usepackage{graphicx}
\usepackage{amssymb}
\usepackage{framed}
\usepackage{subcaption}
\usepackage{xspace}
\usepackage{url}
 \usepackage{bbm} %For indicator variable 1
 \usepackage{float} %To put the figures in here [H]
\usepackage{multirow} % To have multirows in tables
\usepackage{algorithm} %Algorithms
\usepackage{soul} %cross out the text
\usepackage{diagbox}
\usepackage[noend]{algpseudocode} %Algorithms
\usepackage{hyperref} % Make hyper-links
\usepackage{cleveref}
\usepackage{setspace}
\usepackage{boxedminipage}

\hypersetup{
    colorlinks,
    linkcolor={red},%{Mahogany},
    citecolor={blue},
    urlcolor={teal}
}

\ifarxiv
\newtheorem{theorem}{Theorem}

\newtheorem{lemma}[theorem]{Lemma}
\newtheorem{corollary}[theorem]{Corollary}
\newtheorem{definition}[theorem]{Definition}

\fi
\newtheorem{fact}[theorem]{Fact}

\newcommand{\cM}{\mathcal{M}}
\newcommand{\cA}{\mathcal{A}}

\newcommand{\ba}{\mathbf{a}}

\newcommand{\ind}{\mathbf{1}}
\newcommand{\eps}{\epsilon}
\newcommand{\N}{\mathbb{N}}
\newcommand{\Z}{\mathbb{Z}}
\newcommand{\R}{\mathbb{R}}

\newcommand{\bx}{\mathbf{x}}
\newcommand{\bbx}{\bar{x}}

\renewcommand{\phi}{\varphi}

\newcommand{\te}{\tilde{e}}
\newcommand{\tq}{\tilde{q}}
\newcommand{\tS}{\tilde{S}}
\newcommand{\bS}{\bar{S}}
\newcommand{\ds}{\mathbf{S}}
\newcommand{\ods}{\overline{\ds}}
\newcommand{\bA}{\mathbf{A}}
\newcommand{\cQ}{\mathcal{Q}}
\newcommand{\tO}{\tilde{O}}
\newcommand{\tr}{\tilde{r}}
\newcommand{\reach}{\mathrm{CountDistinct}}
\newcommand{\freq}[1]{\mathrm{CntOcc}^{#1}}
\newcommand{\freqhist}[1]{\mathrm{FREQUENCYHIST}^{#1}}

\usepackage[textsize=footnotesize,color=green!40]{todonotes}

\allowdisplaybreaks

\title{Private Counting of Distinct and $k$-Occurring Items in Time Windows}

\ifarxiv
\author{
Badih Ghazi \\
Google Research \\
{\small \texttt{badihghazi@gmail.com}}
\and
Ravi Kumar \\
Google Research \\
{\small \texttt{ravi.k53@gmail.com}}
\and
Pasin Manurangsi \\
Google Research \\
{\small \texttt{pasin@google.com}}
\and
Jelani Nelson \\
UC Berkeley \& Google Research \\
{\small \texttt{minilek@alum.mit.edu}}}
\date{}
\else
\newcommand{\google}{Google Research, Mountain View, CA, USA}
\author{Badih Ghazi}{\google}{badihghazi@google.com}{}{}
\author{Ravi Kumar}{\google}{ravi.k53@gmail.com}{}{}
\author{Jelani Nelson}{UC Berkeley, CA, USA \& \google}{minilek@alum.mit.edu}{}{}
\author{Pasin Manurangsi}{\google}{pasin@google.com}{}{}

\authorrunning{B. Ghazi, R. Kumar, J. Nelson, and P. Manurangsi} 

\Copyright{Badih Ghazi, Ravi Kumar, Jelani Nelson, and Pasin Manurangsi} %TODO mandatory, please use full first names. LIPIcs license is "CC-BY";  http://creativecommons.org/licenses/by/3.0/

 %TODO mandatory: Please choose ACM 2012 classifications from https://dl.acm.org/ccs/ccs_flat.cfm 

 %TODO mandatory; please add comma-separated list of keywords

\relatedversiondetails{Full Version}{https://arxiv.org/abs/???}

\EventEditors{Yael Tauman Kalai}
\EventNoEds{1}
\EventLongTitle{14th Innovations in Theoretical Computer Science Conference (ITCS 2023)}
\EventShortTitle{ITCS 2023}
\EventAcronym{ITCS}
\EventYear{2023}
\EventDate{January 10--13, 2023}
\EventLocation{MIT, Cambridge, Massachusetts, USA}
\EventLogo{}
\SeriesVolume{251}
\ArticleNo{97}
\fi

\begin{document}

\maketitle

\begin{abstract}
In this work, we study the task of estimating the numbers of distinct and $k$-occurring items in a time window under the constraint of differential privacy (DP). We consider several variants depending on whether the queries are on general time windows (between times $t_1$ and $t_2$), or are restricted to being cumulative (between times $1$ and $t_2$), and depending on whether the DP neighboring relation is event-level or the more stringent item-level. We obtain nearly tight upper and lower bounds on the errors of DP algorithms for these problems.
En route, we obtain an event-level DP algorithm for estimating, at each time step, the number of distinct items seen over the last $W$ updates with error polylogarithmic in $W$; this answers an open question of Bolot et al. (ICDT 2013).

% \badih{Regarding the title, not sure (i) whether ``counting distinct elements'' would be more familiar to the NeurIPS community than ``reach'', and (ii) if our use of ``frequency'' will be confusing to the reviewers given that it collides with ``frequency estimation'', which is used in the DP community to mean another functionality, namely estimating the number of occurrences of a given universe element?}

% \badih{For the title, what about something like: ``Cumulative and Time-Window Private Counting of Distinct Elements and $k$-Heavy Hitters''?}

% \badih{Mention that we answer an open question of \cite{BolotFMNT13}.}

% \badih{Algorithms for counting $k$-??? elements.}

\end{abstract}

\ifarxiv
\thispagestyle{empty}
\addtocounter{page}{-1}
\newpage
\fi

\section{Introduction}
Counting distinct elements is a fundamental problem in computer science with numerous applications across different areas including data mining \cite{shukla1996storage, acharya1999aqua,heule2013hyperloglog,padmanabhan2003multi}, computational advertising~\cite{LeckenbyH98,CheongGK10}, computational biology \cite{breitwieser2018krakenuniq, baker2019dashing}, graph analysis \cite{palmer2002anf}, detecting Denial of Service
attacks and port scans \cite{akella2003detecting, estan2003bitmap}, query optimization \cite{selinger1989access, poosala1996improved}. %, and streaming algorithms \cite{kane2010optimal}.
Another related problem
%especially in computational advertising, 
is that of counting {\it $k$-occurring items} in which we wish to count the number of items that occur at least $k$ times in the data for some given parameter $k$.
%(this task is known as computing {\it $k$-frequency} in computational advertising). 
%Furthermore, these counting tasks capture natural use cases across other domains, such as releasing the number of people who have had COVID-19 at least $k$ times or the number of terms queried at least $k$ times in a search engine.
Estimating distinct and $k$-occurring elements over extended time periods capture natural use cases, 
where we wish to study these estimates over particular time windows.
%e.g., estimating the number of people who have had COVID-19 (at least once in the case of  distinct, and, say, exactly $k=2,3$ times in the case of $k$-occurring)  or the number of distinct queries in a search engine.

A concrete application stems from advertising where these problems correspond to the so-called reach and frequency%
\footnote{Here, {\it reach} is the number of individuals (or households) exposed to an ad campaign, whereas the {\it $i$th frequency} $f_i$ denotes the number of individuals exposed to the ad campaign exactly $i$ times \cite{enwiki:1021978492}.} 
of an ad campaign, and are considered two of the most useful metrics \cite{LeckenbyH98,CheongGK10}.  These estimates can then be used to train machine learning (ML) models. For example, it is common for advertisers to seek to optimize the reach and frequency of their campaigns subject to a fixed ad spend budget; they can train an ML model to forecast the reach and frequency histogram for any ad spend budget \cite{quantcast2022}. Constructing the training dataset for this task in turns entails estimating the reach and frequency for a given ad spend budget. Doing so runs the risk of leaking sensitive information about user activity on the publisher sites, which motivates the study of privately estimating reach and frequency \cite{ghazi2022multiparty}. 
%A similar ML use case can be unlocked through (private) estimates of $\reach$ and $\freq{}$ in other use cases as well: e.g., predicting the number of COVID cases or reinfections in different cities.

Increased user awareness and regulatory scrutiny have led to significant research on privacy-preserving algorithms. Differential privacy (DP) \cite{DworkMNS06, dwork2006our} has emerged as a widely popular method for quantifying the privacy of algorithms. Loosely, DP dictates that the output of the (randomized) algorithm remains statistically indistinguishable when a dataset is replaced by a  ``neighboring'' dataset (which differs on the contributions from a single user). 

\begin{definition}[Differential Privacy (DP)~{\cite{DworkMNS06,dwork2006our}}]
For $\eps, \delta \geq 0$, an algorithm $\cM$ is said to be $(\eps, \delta)$-differentially private (i.e. $(\eps, \delta)$-DP) if, for any neighboring input datasets $D, D'$ and any set $O$ of outputs, we have $\Pr[\cM(D) \in O] \leq e^{\epsilon} \cdot \Pr[\cM(D') \in O] + \delta$.
\end{definition}

%The statistical indistinguishability is parameterized by a pair $(\epsilon, \delta)$ of positive real numbers; the smaller their values are, the more private the algorithm is. 
For $\delta = 0$, the algorithm is said to be \emph{pure}-DP, for which we abbreviate $(\eps, \delta)$-DP as just $\eps$-DP. Otherwise, when $\delta > 0$, the algorithm is said to be \emph{approximate}-DP. 

For counting distinct elements, previous work on DP algorithms includes \cite{smith2020flajolet, chen2020distributed, ghazi2022multiparty}, all of which have focused on the case of static datasets, i.e., those on which the queries are evaluated once. However, in practice, datasets that are collected over an extended time period are more prevalent.
More precisely, at each time step $t$, a histogram $S^t$ of universe elements (contributed from multiple users) is added to the dataset. This setting---of estimating queries over a varying dataset and over an extended time period---is the focus of our work. 

%Estimating $\reach$ and $\freq{}$ over extended time periods capture natural use cases across many domains, such as releasing the number of distinct patients in a hospital, the number of people who have had COVID-19 (at least once in the case of $\reach$, and, e.g., exactly $i=1, 2, $ or $3$ times in the case of $\freq{}$)  or the number of distinct terms queried in a search engine.

\subsection{Our Contributions}
We first study the task of counting distinct elements (or ``reach'') over an extended time period. We consider three natural types of queries depending on the underlying time intervals: \emph{time-window} queries (spanning time steps $i$ through $j$ for all $1 \le i \le j \le T$, where $T$ is a given fixed \emph{time horizon} parameter), \emph{cumulative} queries (with the restriction that $i = 1$), and \emph{fixed-window} queries (with the restriction that $j-i = W-1$ for a fixed $W$). %We stress that, throughout the paper, we require the answers to \emph{all} queries to be differentially private. 
We moreover consider the two most natural DP neighboring relations: \emph{item-level} DP (where two datasets are neighboring if they differ on a single universe element's occurrences across \emph{all} time steps), and \emph{event-level} DP (where two datasets are neighboring if they differ on a single universe element's occurrence for a \emph{single} time step).
In addition to the count distinct problem, we also study the closely related task of estimating the number of universe elements that appear at least $k$ times in the input dataset; we call such an item {\em $k$-occurring}. Finally, we consider both the \emph{singleton} setting where a single item arrives at each time step, as well as the \emph{bundle} setting where no such restriction is enforced. For each combination of the previous choices, we prove nearly tight upper and lower bounds on the errors of both pure- and approximate-DP algorithms; see \Cref{ref:summary_table}.

% \badih{Feedback from Pasin: good to mention time-window, cumulative etc. Skip the neighboring notions and the singleton vs bundle setting. State the algorithm in the most general setting it works in. Mention the lower bound and that in fact it works in a more restricted setting that we will get to later. Then, naturally transition to a proof sketch for the subset (and skip the proof sketch for the rest of the combinations)}

\subsection{Organization}
In Section~\ref{sec:setting}, we formally define the setting considered in this paper. We state our results and provide a technical overview of the proofs in Section~\ref{sec:results}. %We discuss relevant related work in Section~\ref{sec:related_work}.
%We present our fixed-window and time-window event-level DP algorithms in Section~\ref{sec:fixed_time_event_algs}.
The next four sections are devoted to the proofs for the case of cumulative, fixed-window, and time-window settings, respectively.
We then conclude with future directions in Section~\ref{sec:conclusion}. 
\ifarxiv
We also discuss additional relevant related work in~\Cref{sec:related_work}.
\fi
%All the missing proofs appear in the Supplementary Material, which also contains additional discussion on related work. 

% \badih{For cumulative reach, the upper bound follows from Bolot et al. Even for cumulative frequency, it should follow as a corollary of one of the generic theorems at the end of their paper; Section 7.2 on `` Holistic Predicate Sum''. Specifically, an analogue of their Lemma 6 should hold for frequency.}

\section{Setting}\label{sec:setting}

Let $[U]$ be a universe of \emph{items}.%
\footnote{Throughout, we write $[m]$ for a non-negative integer $m$ as shorthand for $\{1, \dots, m\}$.}  Let $T$ be the number of time steps (aka, \emph{time horizon}). 
Our input dataset is $\ds = (S^1, \dots, S^T)$ where $S^t$ is a multiset of items at time step $t \in [T]$, i.e., $S^t_j$ denotes the number of occurrences of item $j$ at time step $t$. 
%At each time step $t \in [T]$, the input is a multiset $S^t \in \Z_{\geq 0}^{[U]}$, where $S^t_j$ denotes the number of occurences of item $j$ at time step $t$. 
%Given a function $f: \Z_{\geq 0}^{[d]} \to \R$, our goal is to compute $v_t := f(S_1 + \cdots S_t)$ for all $t \in \N$. (The specific functions of interest are defined in more details below.)
We write $S^{[t_1, t_2]}$ as shorthand for $S^{t_1} + \cdots + S^{t_2}$.

\subsection{Utility Definition}
%\label{sec:util-def}
%
All problems considered in this work can be viewed as a set $\cQ$ of \emph{queries}, where each query $q \in \cQ$ maps each dataset $\ds$ to a real number. The goal of the algorithm is then to output an estimate $\tq_{\bS}$ of $q(\ds)$. Following~\cite{BolotFMNT13}, we say that an algorithm satisfies \emph{$(\alpha, \beta)$-utility} if and only if for every query $q \in \cQ$, $\Pr[|\tq_{\bS} - q(\bS)| \le \alpha] \geq 1 - \beta$.

%we have $\Pr[|e^q_{i, j} - \te^q_{i, j}| \le \alpha] \geq 1 - \beta$. 
For simplicity of exposition, we sometimes only state the bound on $\alpha$ (referred to as the ``error'') for a small constant $\beta$ (e.g., $0.1$). It is simple to see that, in all the cases that we consider, we can get $\beta$ to be arbitrarily small at the cost of at most a $(\log (1/\beta))^{O(1)}$ multiplicative factor in $\alpha$.

\subsection{Cumulative, Time-Window, and Fixed-Window Queries}
Given a function $g: \Z_{\geq 0}^{[U]} \to \R$, we consider three problem versions.
The first is the \emph{time-window} version, which requires us to compute estimates $\te^g_{t_1, t_2}$ of $e^g_{t_1, t_2} := g(S^{[t_1, t_2]})$ for all $1 \leq t_1 \leq t_2 \leq T$. The second is the \emph{cumulative} version, which restricts to the case where $t_1 = 1$. The third is the \emph{fixed-window} version, which is the restriction of the time-window version to the case where $t_2 - t_1 = W - 1$ for some $W \in [T]$ (given beforehand). We view each tuple $(g, t_1, t_2)$ as a query for the utility definition above.
Moreover, the algorithm should be DP when $\te^g_{t_1,t_2}$'s are published simultaneously, for \emph{all} $t_1,t_2$'s of interest.

For simplicity, we assume that the algorithm gets to see all the data (i.e. $S^1, \dots, S^T$) before answering the queries. Another model (also considered in \cite{BolotFMNT13}) is the aforementioned continual release setting where $S^1, \dots, S^T$ arrives in order and we have to answer each query as soon as we have sufficient information to do so (i.e. the algorithm has to answer $e^{g}_{t_1, t_2}$ immediately after receiving $S^{t_2}$). As explained below, our algorithms are obtained through reductions to range query problems; since there are continual release algorithms for the latter~\cite{DworkNPR10,ChanSS11,DworkNRR15}, these imply continual release algorithms for the problems considered in our work too.
We omit the full statements and proofs for brevity.

\subsection{Neighboring Notions}
As we consider DP, we have to specify the neighboring notion under which our algorithm should be DP.  
There are two neighboring relations that we consider in this work:
\begin{itemize}
\item \textit{Item-level DP:} Two datasets $(S^1, \dots, S^T)$ and $(\bS^1, \dots, \bS^T)$ are neighbors iff there exists $u \in [U]$ for which $S^t_{u'} = \bS^t_{u'}$ for all $u' \ne u$ and all $t \in [T]$. In other words, one results from adding/removing a single item $u \in [U]$ to/from (possibly multiple) multisets.

\item \textit{Event-level DP:} Two datasets $(S^1, \dots, S^T)$ and $(\bS^1, \dots, \bS^T)$ are neighbors iff $\sum_{t \in [T]} \|S^t - \bS^t\|_1 \leq 1$. In other words, one results from adding/removing a single element $u \in [U]$ to/from a \emph{single} multiset.
\end{itemize}

We remark that our item-level DP notion coincides with \emph{user-level DP} as in, e.g.,~\cite{DworkNPRY10}.

\subsection{Counting Distinct and $k$-Occurring Elements}
The family of queries we consider is $\freq{\geq k}(S) := |\{u \in [U] \mid S_u \geq k\}|$, i.e., the number of items that appear at least $k$ times in $S$. %Similarly to above, we use $\freq{\geq k}$ to also denote the set of queries that only contains $\freq{\geq k}$.
%There are several (families of) queries that we consider:
%\begin{itemize}
%\item $\reach$ is the function $\reach(S) := |\{u \in [U] \mid S_u > 0\}|$, i.e., the number of distinct items in $S$. We overload the notation and use $\reach$ to also denote the set of queries that only contains $\reach$.
%\item $\freq{\geq k}$ is defined as $\freq{\geq k}(S) := |\{u \in [U] \mid S_u \geq k\}|$, i.e., the number of items that appear at least $k$ times in $S$. Similarly to above, we use $\freq{\geq k}$ to also denote the set of queries that only contains $\freq{\geq k}$.
%ß\end{itemize}

Clearly, $\freq{\geq 1}$ is counting the distinct elements. We also note that it is possible to define $k$-occurring queries that count the number of items that appear \emph{exactly} $k$ times in $S$. It turns out that the bounds in our definitions above carry over with only constant multiplicative overheads. 
\ifarxiv
We provide a more detailed discussion in~\Cref{sec:exact-freq}. 
\else
We omit the details in this version.
\fi

\subsection{Singleton vs Bundle Setting}
Several previous works in the literature (e.g., \cite{DworkNPR10,ChanSS11, BolotFMNT13}) have considered the setting where $\|S_t\|_1 \leq 1$ for all $t \in [T]$. We refer to this setting as the \emph{singleton} setting, whereas the setting where no such restriction is enforced will be referred to as the \emph{bundle} setting.

\section{Our Contributions}\label{sec:results}

\subsection{Highlights of our results}

We derive tight utility bounds in all cases up to polylogarithmic factors in $T$ and the dependency on $k$. %, where errors are measured additively.  
Our results are summarized in Table~\ref{ref:summary_table}. The highlights  of our results are:

\begin{enumerate}
\item {\bf Polylogarithmic errors for all event-level DP problems.} In the event-level DP, we give algorithms for fixed-window and time-window with errors $O(k \log^{1.5} W)$ and $O(k \log^3 T)$ respectively. Note that a cumulative algorithm with $O(\log^{1.5} T)$ error was already presented in~\cite{BolotFMNT13} and it works even in the item-level DP setting. On the other hand, the same work \cite{BolotFMNT13} raised as an open question getting an algorithm with error polylogarithmic in $W$ for the fixed window setting; our aforementioned algorithm answers this question in the affirmative.

\item {\bf Dependency on $k$.} The above upper bound for cumulative does not depend on $k$ whereas those for fixed/time-window grow with $k$. We show that this polynomial dependence on $k$ is necessary.

\item {\bf Lower bounds in other settings.} We show lower bounds of $T^{\Omega(1)}$  and $(T/W)^{\Omega(1)}$ for time-window and fixed-window in the item-level DP setting; we also show matching upper bounds. These lower bounds give separations between event-level and item-level DP.
\end{enumerate}

\setlength{\tabcolsep}{1pt}
\newcommand{\scEvent}{\scriptsize{Event}}
\newcommand{\scItem}{\scriptsize{Item}}
\newcommand{\scLevel}{\scriptsize{Level}}
\newcommand{\scpure}{\scriptsize{pure}}
\newcommand{\scapprox}{\scriptsize{approx}}
\newcommand{\scDP}{\scriptsize{DP}}
\newcommand{\scSingleton}{\scriptsize{Singleton}}
\newcommand{\scBundle}{\scriptsize{Bundle}}

\begin{table*}[htb]
\centering
{\scriptsize
\begin{tabular}{| c | c | c | c | c | c | c | c | c |}
\cline{4-9}
\multicolumn{3}{c |}{} & \multicolumn{2}{ c |}{Cumulative} & \multicolumn{2}{ c |}{Fixed-Window} & \multicolumn{2}{ c |}{Time-Window} \\
\cline{4-9}
\multicolumn{2}{c}{} & Bounds & Upper & Lower & Upper & Lower & Upper & Lower \\
\hline
 & \scpure & \scSingleton & & & & & $O(k \log^3 T)$ &  \\
\cline{3-3}
\scEvent & \scDP & \scBundle & & & $O(k \log^{1.5} W)$ & $\Omega(\sqrt{k})$ & (Cor.~\ref{cor:time-window-alg-pure-concrete}) & $\Omega(\sqrt{k} + \log T)$  \\
\cline{2-3} \cline{8-8}
\scLevel & \scapprox. & \scSingleton & & & (Cor.~\ref{cor:fixed-window-alg-event-level}) & (Lem.~\ref{lem:fixed-window-lb-singleton}, & $O_{\delta}(k \log^2 T)$ & (Cor.~\ref{cor:time-window-lb-polyk-concrete}, \\
\cline{3-3}
& \scDP & \scBundle & & & & Cor.~\ref{cor:fixed-window-lb-polyk-singleton-concrete}) & (Cor.~\ref{cor:time-window-alg-apx-concrete}) & 
Cor.~\ref{cor:time-window-lb-singleton-concrete}) \\
\cline{1-3} \cline{6-9}
 &  & \scSingleton &   &  & & & $\tilde{O}(T)$ & $\Omega(T)$ \\
 & \scpure & & $O(\log^{1.5} T)$ & $-$ & $\tilde{O}(k \cdot \frac{T}{W})$ & $\Omega(\sqrt{k} + \frac{T}{W})$ & (Cor.~\ref{cor:time-window-item-batch-pure}) & (Cor.~\ref{cor:time-window-item-batch-pure-lb}) \\
\cline{3-3} \cline{8-9}
 & \scDP & \scBundle & (Cor.~\ref{cor:cumulative-alg}) & (Lem.~\ref{lem:cum-lb-singleton}) & (Cor.~\ref{cor:fixed-window-alg-item-pure}) & (Cor.~\ref{cor:fixed-window-lb-polyk-singleton-concrete},~\ref{cor:lb-fixed-window-item-pure}) & $\tO_{\delta}(\sqrt{T})$ & $\Omega(\sqrt{T})$  \\
\scItem & & & \cite{BolotFMNT13} & & & & (Cor.~\ref{cor:time-window-item-batch-apx}) & (Cor.~\ref{cor:time-window-item-batch-apx-lb}) \\
\cline{2-3} \cline{6-7} \cline{8-9}
\scLevel & & \scSingleton & & & & & $O(\sqrt{T})$ & $\Omega(\sqrt{T})$ \\
 & \scapprox. & & & & $\tilde{O}_\delta(k \cdot \sqrt{\frac{T}{W}})$ & $\Omega(\sqrt{k} + \sqrt{\frac{T}{W}})$ & (Cor.~\ref{cor:time-window-item-singleton-pure}) & (Cor.~\ref{cor:time-window-item-singleton-pure-lb}) \\
\cline{3-3} \cline{8-9}
& \scDP & \scBundle & & & (Cor.~\ref{cor:fixed-window-alg-item-apx}) & (Cor.~\ref{cor:fixed-window-lb-polyk-singleton-concrete},~\ref{cor:lb-fixed-window-item-apx}) & $\tO_{\delta}(\sqrt[3]{T})$ & $\Omega(\sqrt[3]{T})$ \\
& & & & & & & (Cor.~\ref{cor:time-window-item-singleton-apx}) & (Cor.~\ref{cor:time-window-item-singleton-apx-lb}) \\
\hline
\end{tabular}
}
\caption{Summary of our results. We use $\tO(f)$ to hide factors polylogarithmic in $f$. For simplicity of the bounds, we assume that $\eps, \beta \in (0, 1]$ are constants, $T > W^2$, and $W > k^2$.
Note that we also provide reductions from 1d-range query to cumulative and time-window $\freq{\geq k}$ in event-level DP. However, no concrete bound is shown in the table because, to the best of our knowledge, no non-trivial bound for 1d-range query is known for the utility notion we use, although an $\Omega\left(\frac{\log M}{\eps}\right)$ lower bound is known for the stronger $\ell_\infty$ utility notion (see, e.g.,~\cite{vadhan2017complexity}).
}\label{ref:summary_table}
\end{table*}

We next proceed to describe the main ideas in our algorithms and reductions. For ease of presentation, we group the results together based on proof techniques and only explain the high-level approaches. All omitted details will be formalized later.

\subsection{Algorithmic Overview: Event-Level DP Setting}
\label{sec:tech-overview}

We start our algorithmic overview with the \emph{event-level} DP setting, which is our main contribution. Our results will establish reductions to the range query problem.

\begin{definition}[Range Query] \label{def:range-queries}
In the \emph{$d$-dimensional range query} problem, the input consists of $x^1, \dots, x^n \in \{0, \dots, M\}^d$. The goal is to output, for every $y^1, y^2 \in \{0, \dots, M\}^d$, $r_{y^1, y^2}(\bx) := |\{j \in [n] \mid y^1 \preceq x^j \preceq y^2\}|$, where $x \preceq y$ iff $x_i \leq y_i$ for all $i \in [d]$.
\end{definition}
Here DP is with respect to adding/removing a single $x^j$.  DP algorithms for range query with errors polylogarithmic in $M$ are known (see \Cref{sec:prelim-range-queries}.) Indeed, DP range query remain an active area of research (e.g.,~\cite{HenzingerU22,McMahanRT22}); by reducing to/from range query, any future improvements there can be immediately combined with our reductions to yield better bounds for our problems as well.
%
%Recall that in the $d$-dimensional range query problem. Each user $j$ has a point $x^j \in \{0, \dots, M\}^d$. The goal is then to output, for every $y^1, y^2 \in [M]^d$, the estimate of $r_{y^1, y^2}(\bx) := |\{j \in [n] \mid y^1 \preceq x^j \preceq y^2\}|$, where $\preceq$ denotes all-coordinate partial ordering. 

To describe the reductions, let us also denote by $t^u_\ell$ the time step that an item $u \in [U]$ is reached for the $\ell$th time\footnote{Note that in the actual algorithm, we need to handle the case where $S^t_u > 1$ and also handle the boundaries. For a formal treatment of the modification required for handling this, see \Cref{sec:prelims}.}.
The reductions in each case proceed as follows:

\subsubsection{Cumulative} The reduction was already implicit in~\cite{BolotFMNT13}, but we present it in our framework as a warm-up to the other results. It is simple to see that $u$ should be counted in the $(1, t)$ query (i.e., the cumulative query ending at time $t$) if and only if $t^u_k \leq t$. The reduction is now clear: for each item $u$, we create a point $t^u_k \in [T]$. Running the 1d-range query algorithm on these points then gives us estimates for the cumulative count distinct values.

\subsubsection{Time-window} First we focus on the $k = 1$ case. Observe that $u$ should \emph{not} be included in the $(t_1, t_2)$ query iff for some $\ell$, it holds that $t^u_\ell < t_1, t_2 < t^u_{\ell + 1}$. This leads to the following reduction to the 2d-range query problem. For each item $u \in [U]$, we create points $(t^u_{\ell} + 1, t^u_{\ell+1} - 1)$ for all $\ell$. To answer a $(t_1, t_2)$ time-window query, we use a range query to find the number of points whose\footnote{This corresponds to a range query in \Cref{def:range-queries} with $y^1 = (-\infty, t_2), y^2 = (t_1, \infty)$.} first coordinate is at most $t_1$ and second coordinate is at least $t_2$, which gives the number of items \emph{not} included in the desired time query.
%\jelani{This is written in a slightly different language from the defn of `range query' given in Definition 2; in that defn we count the number of input points ``in between'' two query points, but what's in this paragraph is a 2-sided range query. I guess the point is this does match Defn 2 by setting $y^1 = (-\infty, t_2), y^2 = (t_1, \infty)$.}
%\pasin{That's right; added a footnote to clarify this.} We then subtract this number from $U$ to get the answer.

It is important to take the ``complement'' perspective in the above reduction. Specifically, if we were to use the fact that $u$ should be included iff $t_1 \leq t^u_\ell \leq t_2$ for some $\ell$, this will lead to double counting because the condition can be satisfied by multiple values of $\ell$. However, the condition given in the previous paragraph is satisfied by (at most) a single value of $\ell$. 

For the $k > 1$ case, we can use a similar approach but change the condition to $t^u_\ell < t_1, t_2 < t^u_{\ell + k}$ (i.e., replacing $\ell + 1$ with $\ell + k$). However, this leads to double counting because it is possible that $t_{\ell + k} > t_2$ and $t_{\ell}, t_{\ell + 1}$ are both less than $t_1$. To solve this issue, we create another instance of 2d-range query with respect to the condition $t^u_{\ell + 1} < t_1, t_2 < t^u_{\ell + k}$. We then subtract the answer of the second instance from that of the first instance; it can be seen that this results in no double counting. Finally, note that a single event change results in at most $O(k)$ changes to the instances because at most $k+1$ conditions $t^u_\ell < t_1, t_2 < t^u_{\ell + k}$ are affected. By group privacy (\Cref{fact:group-dp}),
%\jelani{By this point the defn of `group privacy' hasn't been given yet.}
%\pasin{Added a forward pointer}, 
this means we may apply the $(O(\eps/k), O(\delta/k))$-DP 2d-range query algorithm and still get $(\eps, \delta)$-DP in our answer. %\badih{Should $\epsilon$ appear in the $O(\delta/k))$ term, when applying group privacy above?}
%\pasin{Not really, assuming that $\eps \leq 1$.}

\subsubsection{Fixed-window}  In the fixed-window case, we can do better than time-window by reducing to 1d- rather than 2d-range query. Again, let us start from the $k = 1$ case. Recall from the above that an item $u \in [U]$ should \emph{not} be included in the fixed-window query $(t, t + W - 1)$ iff $t^u_{\ell} < t, t + W - 1 < t^u_{\ell + 1}$. This condition can be further simplified as $t \in (t^u_\ell, t^u_{\ell + 1} - W]$. Therefore, we may create two instances of the 1d-range query problem: if $t^u_{\ell+1} - t^u_{\ell} > W$, add $t^u_\ell$ to the first instance and $t^u_{\ell + 1} - W$ to the second instance. To count the number of items $u$ such that $t \in (t^u_\ell, t^u_{\ell + 1} - W]$, we query for the points of value at most $t$ in each instance and then subtract the first answer from the second answer. This gives us the number of items that should not be included, which we subtract from $U$.

For $k>1$, we start from the condition $t \in (t^u_\ell, t^u_{\ell + k} - W]$. Naturally, this gives us the following algorithm: if $t^u_{\ell+k} - t^u_{\ell} > W$, add $t^u_\ell$ to the first instance and $t^u_{\ell + k} - W$ to the second instance. (Then, proceed as above.) Unfortunately, this again results in double counting when $t^u_{\ell + k} \geq t + W$ and $t^u_\ell, t^u_{\ell+1}$ are both less than $t$. To resolve this, we change the second point from $t^u_{\ell + k} - W$ to $\min\{t_{\ell + 1}, t^u_{\ell + k} - W\}$. %It can be verified that this resolves the double counting issue.

The above ideas are sufficient for a bound of the form $k \cdot \log^{O(1)} T$. To decrease $\log T$ to $\log W$, we make the following observation: since the window length is fixed to $W$, we may run the above algorithm for time steps $1, \dots, 2W$, then run it again for $W + 1, \dots, 3W$, then for $2W + 1, \dots, 4W$ and so on. This suffices to answer any fixed-window query of length $W$. Furthermore, observe that each event-level change only affects two runs, meaning that we only need each run to be $(\eps/2, \delta/2)$-DP. Since each run now has time horizon only $2W$, we have decreased $\log T$  to $\log W$ as desired.

\subsection{Algorithmic Overview: Item-Level DP Setting}

We next briefly describe how to extend each result to the item-level DP case. First, we remark that the above cumulative algorithm actually works also with item-level DP because each item contributes to at most one point in the 1d-range query instance. As for the time-window and fixed-window settings, they are solved by reducing to multiple calls of easier problems and using composition theorems:

\subsubsection{\bf Time-window} We run the cumulative algorithm starting at each time step. Via the basic or advanced composition theorems (\Cref{thm:basic-comp,thm:adv-comp}),
%\jelani{perhaps forward reference to where these are defined}
%\pasin{Added forward pointers}, 
it suffices for each cumulative algorithm to be $(\eps/T)$-DP and $O(\eps/\sqrt{T\log(1/\delta)})$-DP respectively, which yields the nearly tight bounds.

\subsubsection{\bf Fixed-window} This case is more subtle. We start by observing that the above event-level DP algorithm for $T = 2W$ in fact works even in the item-level DP setting: the reason is that the condition $t^u_{\ell + k} - t^u_{\ell} > W$ cannot occur for two different indices $\ell$ that are at least $k$ apart. (Otherwise, we would have $2W > 2W$, a contradiction.) This means that we can still solve each subinstance of time horizon $2W$ with error $k \cdot \log^{O(1)} W$. Since there are $O(T/W)$ subinstances (in the reduction described above), we can use the basic/advanced composition over them. This gives us the final $(T/W)^{O(1)} \cdot \log^{O(1)} W$ bounds.

Finally, we remark that we can sometimes obtain improvements in the singleton setting. We do so via a technique from~\cite{jain2021price} where the $T$ time steps are partitioned into $\lceil T/T' \rceil$ grouped time steps, each consisting of $\leq T'$ consecutive time steps. The main observation is that we may run the bundle algorithm on the $\lceil T/T' \rceil$ grouped time steps and incur an additional error of at most $T'$ due to the grouping. By picking $T'$ appropriately, this yields the bounds for our time-window algorithms in the singleton setting.

\subsection{Overview of the Lower Bounds}
At a high-level, we prove two types of lower bounds. First are the lower bounds from 1d-/2d-range query. These are done by essentially ``reversing'' our algorithms explained above. In fact, these reductions show that the cumulative, fixed-window, and time-window $\freq{\geq k}$ are equivalent to 1d-, 1d-, and 2d-range query, respectively, up to a multiplicative factor of $O(k)$ in the values of $\eps, \delta$.

The second type of reduction is from (variants of) the 1-way marginal problem. Recall that in the \emph{1-way marginal problem} each user receives $x^j \in \{0, 1\}^d$ and the goal is to compute $\bbx := \sum_{j \in [n]} x^j$. Lower bounds of $d^{\Omega(1)}$ are well-known for this problem~\cite{HardtT10,BunUV18}.

\subsubsection{Polynomial in $T$ and $T/W$ Lower Bounds} Let us start with the simpler reductions for lower bounds of the form $(T/W)^{\Omega(1)}$ and $(T)^{\Omega(1)}$ for fixed-window and time-window queries in item-level DP. %These are much simpler than the $k^{\Omega(1)}$ lower bound, as we are now in the item-level setting. 
For the ease of presentation, let us consider the time-window bundle setting with $k = 1$ as an example. In this case, we let $d = T, U = n$, and $S^j_i = x^j_i$ for all $j \in [n], i \in [d]$ (i.e., ``copying'' each vector $x^j$ into the occurrence pattern of item $j$). It can be immediately seen that $\freq{\geq 1}(S^i)$ is exactly equal to $\bbx_i$, which translates the lower bounds of $d^{\Omega(1)} = T^{\Omega(1)}$ from 1-way marginals to time-window bundle $\freq{\geq 1}$. 

We remark that our reductions bear  similarities to those in~\cite{jain2021price}, although, given that different problems are studied in that work, the reductions and arguments are not the same.

\subsubsection{Polynomial in $k$ Lower Bounds} Lower bounds of the form $k^{\Omega(1)}$ are also shown via reductions from (variants of) 1-way marginals with $d = \Theta(k)$. However, these are more challenging than the above as we have to handle the event-level DP setting. Notice that in the above reduction each user corresponds to a vector $x^j \in \{0, 1\}^d$ and changing this vector can cause the change in $S^j_i$ for all $i \in [T]$, which invalidates its use in the event-level DP setting. For event-level DP reductions, we have to somehow be able to encode the entire information of the $d$-dimensional vector $x^j \in \{0, 1\}^d$ while ensuring that changing this entire vector only causes a single event change.

The aforementioned challenges led us to the idea of creating multiple items for each possible value of $x^j$. In other words, we now create an item for each $(j, x) \in [n] \times \{0, 1\}^d$. Each of these items is ``dormant'', meaning that it does not contribute to $\freq{\geq k}$ queries at all. However, each item $(j, x)$ can be ``turned on'' by a single event change, after which it contributes to the queries in a pattern designated by the vector $x$. This is done by ensuring that for each $i \in [d]$ there are prespecified $t_i, t_i'$ such that $S^{[t_i, t'_i]}_{(j, x)}$ is equal to $k - 1 + x_i$. The event-level change to turn on $(j, x)$ is to increase $S_t^{(j, x)}$ such that $t \in [t_i, t'_i]$ for all $i \in [d]$. (Note that, since we are aiming for event-level change, we need to use a single $t$ that satisfies $t \in [t_i, t'_i]$ \emph{simultaneously} for all $i \in [d]$.) %\badih{need to add a sentence saying the intersection is not empty} 
After turning on $(j, x^j)$ for all $j \in [n]$, we thus have $\bbx_i = \freq{\geq k}(S^{[t_i, t'_i]})$, allowing us to compute the 1-way marginal. This concludes the main idea of the reduction.

While the sketch above suffices for the bundle setting, it is not enough for the singleton case. This is because we created as many as $2^d = 2^{O(k)}$ dormant items, which requires $T > 2^d$ in the singleton case. This is too large to yield any meaningful results. Instead, we turn to the so-called \emph{linear queries} problem, which may be viewed as a restriction of 1-way marginals where each $x_j$ belongs to a (predefined) small set of size $d^{O(1)}$. (See \Cref{subsec:linear-q} for a formal definition.) A classic result of~\cite{DinurN03} showed that, even in this restricted version, any DP algorithm has to incur an error of $d^{\Omega(1)}$. Since we now only have to create dormant items $(j, x)$ for $x$ that belongs to the set, we can now let $T$ be as small as $d^{O(1)}$.

\subsection{Related Work}
\label{sec:related_work}

%The most relevant previous work in this line of research is that of Bolot et al.~\cite{BolotFMNT13} who, building on top of the aforementioned algorithms from~\cite{DworkNPR10,ChanSS11}, provided an algorithm for cumulative count distinct with error polylogarithmic in $T$. %To the best of our knowledge, this is the only previous positive result on counting distinct elements in time windows.
%
%Moreover, Dwork et al. \cite{DworkNPR10} showed that a logarithmic error is needed. The binary tree mechanism was extended to the task of estimating weighted sums with exponentially decaying coefficients \cite{BolotFMNT13}, and sums of bounded real values \cite{perrier2018private}. The continuous release model of DP was later studied on graph data and statistics in \cite{song2018differentially,fichtenberger2021differentially}, and the aforementioned DP binary tree mechanism was used in the context of DP online learning \cite{jain2012differentially,guha2013nearly,agarwal2017price}. Very recently, Jain et al. \cite{jain2021price} show that tasks closely related to summation, namely requiring the selection of the largest of a set of sums (and previously studied by \cite{cardoso2022differentially}), have to incur an error polynomial in $T$ in the continuous release model, and are fundamentally harder than in the standard (bundle) setting. \badih{Add a sentence about why the continuous release model doesn't fully capture our setting.} 

A closely related line of work is the continual release model of DP, which was first introduced in the work of Dwork et al. \cite{DworkNPR10}, and Chan et al.\cite{ChanSS11}. For the binary summation task, they designed a DP mechanism in this model, the binary tree mechanism, which achieves an additive error polylogarithmic in $T$; Dwork et al.~\cite{DworkNPR10} also showed that a logarithmic error is needed. The binary tree mechanism was extended to the task of estimating weighted sums with exponentially decaying coefficients \cite{BolotFMNT13}, and sums of bounded real values \cite{perrier2018private}. The continuous release model of DP was later studied on graph data and statistics in \cite{song2018differentially,fichtenberger2021differentially}, and the aforementioned DP binary tree mechanism was used in the context of DP online learning \cite{jain2012differentially,guha2013nearly,agarwal2017price}. Very recently, Jain et al. \cite{jain2021price} show that tasks closely related to summation, namely requiring the selection of the largest of a set of sums (and previously studied by \cite{cardoso2022differentially}), have to incur an error polynomial in $T$ in the continuous release model, and are fundamentally harder than in the standard setting. We remark that this separation implies a polynomial separation between cumulative queries and a single query (for their problems) in our terminologies. This is a different behavior compared to the $\reach$ and $\freq{}$ problem studied in our work, as there is a polylogarithmic-error algorithm (\Cref{lem:cumulative-alg}).  In fact, we use Jain et al.'s~\cite{jain2021price} technique in our upper bound for the item-level setting. Our lower bound reductions bear some similarities to theirs, but there are some fundamental differences. For example, we crucially use the structure of $\freq{}$ when proving $k^{\Omega(1)}$ lower bounds.

% \badih{Add a sentence about why the continuous release model doesn't fully capture our setting.}

% \badih{Cite: ``The Price of Differential Privacy under Continual Observation''. They have a less natural problem than we do. Will need to check if we get a bigger separation.} 

$\reach$ and $\freq{}$ have also been studied in the more challenging pan-private model~\cite{DworkNPRY10,MirMNW11}, where the DP constraint is applied not only to the output but also to the internal memory of the algorithm (after any time step); hence, the utility guarantees are much weaker than the algorithms we present (and those in \cite{BolotFMNT13}). Mir et al.~\cite{MirMNW11} proves a strong lower bound in this model: any pan-private algorithm for $\reach$ must incur an additive error of at least $\Omega_\eps(\sqrt{U})$. Due to such a lower bound,~\cite{DworkNPRY10,MirMNW11} design several algorithms for $\reach$ and $\freq{}$ with guarantees that include both \emph{multiplicative} and additive errors. This is in contrast to our algorithms, which do not require any multiplicative errors and have additive error bounds that are completely independent of $U$. 

We also note that $\reach$ has been studied in the local, and shuffle settings of DP (see, \cite{balcer2021connecting, chen2020distributed}), although these results focus on the single-query (i.e. $T = 1$) setting.

\section{Preliminaries} \label{sec:prelims}

In the algorithms below, we will assume that the input is appended at the beginning and at the end with multisets $S^0, S^{T+1}$ such that $S^0_u = S^{T+1}_u = k$ for all $u \in [U]$. This helps simplify the notation in the reductions below. For each item $u \in [U]$, it is also convenient to define the sequence $0 = t^1_u \leq \cdots \leq t^{m_u}_u = T + 1$ of time steps $u$ is reached; this is the sequence where each $t \in \{0, 1, \dots, T + 1\}$ appears $S^t_u$ times.
We use $S^{\leq t}$ for $S^1 + \cdots + S^t$.

\subsection{Tools from Differential Privacy}

We list below several tools that we will need from the DP literature.
We refer the reader to the monographs \cite{dwork2014algorithmic, vadhan2017complexity} for a thorough overview of DP.

\subsubsection{Composition Theorems}
We will use the following composition theorems throughout this paper.

\begin{theorem}[Basic Composition] \label{thm:basic-comp}
For any $\eps, \delta > 0, k \in \N$,
an algorithm that is a result of running $k$ mechanisms, each of which is $(\eps/k, \delta/k)$-DP, is $(\eps, \delta)$-DP.
\end{theorem}

% \badih{Given the assumption in the next theorem, are we somehow restricting $\eps$ to be $<1$ in some of the results in the paper? Feedback from Pasin: yes, let's state somewhere (not necessary to do so in all theorems) that we're going with the restriction on $\eps$ for simplicity. We can only claim closeness to optimality for small $\eps$.  Ravi: There is a sentence around line 603 ptal.}

\begin{theorem}[Advanced Composition~\cite{DworkRV10, dwork2014algorithmic}] \label{thm:adv-comp}
For any $\eps, \delta \in (0, 1), k \in \N$,
an algorithm that is a result of running $k$ mechanisms, each of which is $\left(\frac{\eps}{2\sqrt{2k\ln(2/\delta)}}, \frac{\delta}{2k}\right)$-DP, is $(\eps, \delta)$-DP.
\end{theorem}

\begin{theorem}[Parallel Composition~\cite{McSherry10}] \label{thm:parallel-comp}
Let $\phi_1, \dots, \phi_t$ be deterministic functions that map a dataset to another dataset such that, for any neighboring datasets $D, D'$, there exists $j \in [t]$ for which $\phi_i(D) = \phi_i(D')$ for all $i \ne j$ and $\phi_j(D), \phi_j(D')$ are neighbors. Then, an algorithm that is a result of running an $(\eps, \delta)$-DP algorithm on each of $\phi_1(D), \dots, \phi_t(D)$ is $(\eps, \delta)$-DP.
\end{theorem}

\subsubsection{Group Privacy} Let $\sim$ denote any neighboring relationship. For any $k \in \N$, let $\sim_k$ denote the neighboring relationship where $D \sim_k D'$ if and only if there exists $D = D_1, \dots, D_k = D'$ such that $D_i \sim D_{i + 1}$ for all $i \in [k - 1]$. The following so-called \emph{group privacy} bound is well-known and simple to prove.

\begin{fact}[Group Privacy (e.g.,~\cite{SteinkeU16})] \label{fact:group-dp}
Let $\eps, \delta > 0$ and $k \in \N$. Any algorithm $\cM$ that is $(\eps, \delta)$-DP under some neighboring relationship $\sim$ is $\left(k\eps, \frac{e^{k\eps}-1}{e^{\eps}-1}\delta\right)$-DP under the neighboring relationship $\sim_k$.
\end{fact}

\subsection{Range Query}
\label{sec:prelim-range-queries}

Here, we view each $r_{y^1, y^2}$ as a query, and the definition of utility is the same as in~\Cref{sec:setting}. The following algorithms for range query are known from previous works:

\begin{theorem}[\cite{DworkNPR10,ChanSS11}] \label{thm:1d-pure}
For any $\eps > 0$, there is an $\eps$-DP algorithm for the 1d-range query problem with $O(\log^{1.5} M \cdot \log(1/\beta) / \eps, \beta)$-utility.
\end{theorem}

\begin{theorem}[\cite{DworkNRR15}] \label{thm:2d-pure}
For any $\eps > 0$, there is an $\eps$-DP algorithm for the 2d-range query problem with $O(\log^3 M \cdot \log(1/\beta) / \eps, \beta)$-utility.
\end{theorem}

Although the following was not stated explicitly in~\cite{DworkNRR15}, it is not hard to derive, by simply replacing Laplace noise by Gaussian noise.

\begin{theorem}[\cite{DworkNRR15}] \label{thm:2d-apx}
For any $\eps, \delta \in (0, 1)$, there is an $(\eps, \delta)$-DP algorithm for the 2d-range query problem with $O(\log^2 M \cdot \sqrt{\log(1/\delta) \log(1/\beta)} / \eps, \beta)$-utility.
\end{theorem}

As for the lower bound, we are not aware of any non-trivial lower bound for the 1d case for the utility notion we use, although lower bounds for $\ell_\infty$-error are known (see e.g.~\cite{vadhan2017complexity}). 
The lower bound for the 2d case is stated below.

\begin{theorem}[\cite{MuthukrishnanN12}] \label{thm:2d-range-lb}
For any sufficiently small constants $\eps, \delta, \beta > 0$, there is no $(\eps, \delta)$-DP algorithm with $(o(\log M), \beta)$-utility for the $2$d-range query problem even when the number of input points $n$ is $O(M^2)$.
\end{theorem}

\subsection{1-Way Marginal}

In the 1-way marginal problem, we are given $x^1, \dots, x^n \in \{0, 1\}^d$, and the goal is to output an estimate of $\bbx = x^1 + \cdots + x^n$. Here we view each $\bbx_j$ as a query, and the definition of utility is in~\Cref{sec:setting}.

\begin{theorem}[\cite{HardtT10}] \label{thm:marginal-pure-lb}
For any $\eps \in (0, 1]$, there is no $\eps$-DP algorithm with $(o(\min\{n, d/\eps\}), 0.01)$-utility for the $1$-way marginal problem.
\end{theorem}

\begin{theorem}[\cite{SteinkeU16}] \label{thm:marginal-apx-lb}
For any $\eps \in (0, 1)$ and $\delta \in (2^{-\Omega(n)}, 1/n^{1+\Omega(1)})$, there is no $(\eps, \delta)$-DP algorithm with $(o(\min\{n, \sqrt{d \log(1/\delta)}/\eps\}, 0.01)$-utility for the $1$-way marginal problem.
\end{theorem}

\subsection{Linear Queries}
\label{subsec:linear-q}

The linear query problem is parameterized by a set of (public) vectors $\ba_1, \dots, \ba_d \in \{0, 1\}^m$. The input to the algorithm is then a vector $\bx \in \{0, 1\}^m$ and the goal is to estimate $\left<\ba_i, \bx\right>$ for each $i \in [d]$. We view each $\ba_i$ as a query and utility is defined as in~\Cref{sec:setting}.

The neighboring relation for DP of linear queries is with respect to changing a single coordinate of $\bx$. For convenience, we may also think of the queries as the matrix $\bA \in \{0, 1\}^{d \times m}$ where $\ba_i$ is the $i$th row of $\bA$. Under this notation, we  have $(\bA\bx)_i = \left<\ba_i, \bx\right>$. We write ``$\bA$-linear query'' to signify the matrix $\bA$.

The following lower bound was shown in \cite{DworkMT07} and was an improvement over the classical lower bound of Dinur and Nissim~\cite{DinurN03}.

\begin{theorem}[\cite{DworkMT07}] \label{thm:linear-query-lb}
For any $m \in \N$ and any sufficiently small constants $\eps, \delta > 0$, there exists a matrix $\bA \in \{0, 1\}^{d \times m}$ with $d = O(m)$ such that there is no $(\eps, \delta)$-DP algorithm with $(o(\sqrt{m}), 0.01)$-utility for the $\bA$-linear query problem.
\end{theorem}

\subsection{Assumptions on Parameters}
\label{sec:parameter-assumption}

Certain settings of parameters can lead to ``degenerate'' cases, for which there are uninformative algorithms that can get better errors. To avoid such scenarios, we will assume the following setting of parameters throughout the paper when we prove our lower bounds:
\begin{itemize}
\item $T, W \geq k \log k$. This is due to the fact that, in the singleton setting, the algorithm that outputs zero always gets an error of at most $T/k$ or $W/k$, which can be smaller than meaningful algorithms when $T, W$ are not much larger than $k$. 
\item $T \geq (1 + \Omega(1))W$. This is due to the fact that there are only $T - W + 1$ queries in the fixed-window setting, meaning that even, e.g., the Laplace mechanism would yield an error of $O_\eps(T - W)$. This can be small if $T - W$ is very small.
\end{itemize}
We stress that our upper bounds work for all settings of parameters (even those violating the above assumptions), but we assume the above for simplicity in the lower bound statements.

For simplicity of utility expressions, we assume throughout that $\eps \in (0, 1], \delta \in [0, 1/2)$ (including both in the upper and lower bounds). Our results can be extended to the $\eps \gg 1$ case but the utility expressions are more complicated because, e.g., the advanced composition theorem~\cite{DworkRV10} has a more complicated expression in this case.

\section{Cumulative Queries}

We are now ready to prove our results, starting with algorithms and lower bounds for cumulative $\freq{\geq k}$.

\subsection{Algorithm}

%We start with the algorithm for cumulative $\freq{\geq k}$. 
As stated earlier, a similar algorithm was already derived in~\cite{BolotFMNT13} but with less emphasis on the relationship with 1d-range query. We provide a proof below both for completeness and for providing a formal relationship with 1d-range query.

\begin{lemma} \label{lem:cumulative-alg}
Let $k \in \N$ be any positive integer.
If there exists an $(\eps, \delta)$-DP algorithm for 1d-range query with utility $(\alpha(M, \eps, \delta, \beta), \beta)$, then there exists an $(\eps, \delta)$-DP algorithm for cumulative $\freq{\geq k}$ in the item-level DP and bundle setting with utility $(\alpha(T + 1, \eps, \delta, \beta), \beta)$.
\end{lemma}

\begin{proof}
Let $M = T + 1$ and $n = d$. For each item $u \in [d]$, we then define %$x_u := \min \{t \mid S^{\leq t}_u \geq k\}$
$x_u := t_u^k$, i.e., the first time step before which (inclusive) $u$ has appeared $k$ times. It is not hard to see that the prefix query $r_{1, t}$ is exactly equal to $\freq{\geq k}(S^{\leq t})$. Therefore, we can run the $(\eps, \delta)$-DP algorithm for 1d-range query, which yields the desired error. Finally, notice that each item contributes to only one input point to the 1d-range query problem; therefore, the algorithm remains $(\eps, \delta)$-DP under the item-level neighboring notion.
\end{proof}

Plugging the above into known algorithm for 1d-range query (\Cref{thm:1d-pure}) yields:
\begin{corollary} \label{cor:cumulative-alg}
For any $k \in \N$ and $\eps > 0$, there is an $\eps$-DP algorithm for cumulative $\freq{\geq k}$ in the item-level DP and bundle setting with $(O(\log^{1.5}T\log(1/\beta)/\eps, \beta)$-utility.
\end{corollary}

\subsection{Lower Bounds}

We now prove a lower bound showing a reverse reduction---from 1d-range query to cumulative $\freq{\geq k}$---complementing our algorithm in~\Cref{lem:cumulative-alg}. This shows that the two problems are equivalent (up to a constant factor in the error). We remark that our lower bounds below hold even in the more stringent event-level DP setting.

We start with a slightly simpler reduction in the bundle setting:

\begin{lemma} \label{lem:cum-lb-batch}
Let $k$ be any positive integer and $\eps, \delta > 0$.
If there exists an $(\eps, \delta)$-DP algorithm for cumulative $\freq{\geq k}$ in the event-level DP and bundle setting with $(\alpha(T, \beta), \beta)$-utility, then there exists an $(\eps, \delta)$-DP algorithm for 1d-range query with $(2 \cdot \alpha(M, \beta / 2), \beta)$-utility.
\end{lemma}
\begin{proof}
Let $T = M$ and let $U$ be sufficiently large (i.e., larger than the dataset size of the 1d-range query). Given an input $x^1, \dots, x^n \in [M]$ to the 1d-range query problem, we construct an input to the cumulative $\freq{\geq k}$ algorithm as follows:
\begin{itemize}
\item For all $u \in [U]$, let $S^1_u = k - 1$ and $S^2_u = \cdots = S^T_u = 0$.
\item For all $j \in [n]$, increment $S^{x^j}_{j}$ by one.
\end{itemize}
It is simple to see that, for all $t \in [T]$, $\freq{\geq k}(S^{\leq t}) = r_{0, t}(\bx)$. Therefore, we may answer any range query $r_{y_1, y_2}(\bx)$ by outputting $\freq{\geq k}(S^{\leq y_2}) - \freq{\geq k}(S^{\leq y_1 - 1})$. This is indeed an $(\eps, \delta)$-DP algorithm for 1d-range query with $(2 \cdot \alpha(M, \beta/2), \beta)$-utility.
\end{proof}

In the singleton setting, we cannot use the above reduction directly since the first step in the previous reduction contains $k - 1$ copies of $u$'s at the same time step. Therefore, we have to ``expand'' this set into $(k - 1)n$ sets where $n$ denote the number of input points in the 1d-range query problem. This results in a slight additive blow up of $(k - 1)n$ in the time horizon. Similarly, the second step in the reduction can increment multiple values at the same time step; to avoid this, we have to pay another multiplicative blow up of $k - 1$. These are formalized below. %We remark that this $k - 1$ additive factor is necessary: if the time  %, although as we explain below this does not effect the concrete bounds we eventually get.

\begin{lemma} \label{lem:cum-lb-singleton}
Let $k$ be any positive integer and $\eps, \delta > 0$.
If there exists an $(\eps, \delta)$-DP algorithm for cumulative $\freq{\geq k}$ in the event-level DP and singleton setting with $(\alpha(T, \beta), \beta)$-utility, then there exists an $(\eps, \delta)$-DP algorithm for 1d-range query on at most $n$ input points with $(2 \cdot \alpha((M + k - 1) n, \beta/2), \beta)$-utility.
\end{lemma}

\begin{proof}
Let $T = (M + k - 1)n$ and let $U = n$. Given an input $x^1, \dots, x^{n'} \in [M]$ to the 1d-range query problem where $n' \leq n$, we construct an input to the cumulative $\freq{\geq k}$ algorithm as follows:
\begin{itemize}
\item For all $u \in [U]$, let $S^{(k-1)(u-1) + 1}_u = \cdots = S^{(k-1)u}_u = 1$ and $S^t_u = 0$ for all $t \in [T] \setminus \{(k-1)(u-1) + 1, \dots, (k - 1)u\}$.
\item For all $j \in [n]$, increment $S^{n(k - 1) + n(x^j -1) + j}_{j}$ by one.
\end{itemize}
It is not hard to see that this is a singleton instance.
Similar to before, we can answer any range query $r_{y_1, y_2}(\bx)$ by outputting $\freq{\geq k}(S^{\leq n(y_2  + k -1)}) - \freq{\geq k}(S^{\leq n(y_1  + k - 2)})$. This gives an $(\eps, \delta)$-DP algorithm for 1d-range query with $(2 \cdot \alpha((M + k - 1) n, \beta/2), \beta)$-utility.
\end{proof}

\section{Fixed-Window Queries}

We next move on to prove our bounds for fixed-window queries. Since the bounds are different for the two DP notions, we start with event-level DP and then consider item-level DP.

\subsection{Event-Level DP}

\subsubsection{Time Horizon Reduction: Proof of \Cref{lem:horizon-reduction-fixed-window}}

We start with a lemma that allows us to reduce the time horizon $T$ to $2W$ in this setting. The proof of this lemma follows the overview discussed in \Cref{sec:tech-overview}.

\begin{lemma} \label{lem:horizon-reduction-fixed-window}
For any $k \in \N$, if there exists an $(\eps, \delta)$-DP algorithm for fixed-window $\freq{\geq k}$ with $(\alpha(T, \eps, \delta, \beta), \beta)$-utility with event-level DP, there exists an $(\eps, \delta)$-DP algorithm with $(\alpha(2W, \eps/2, \delta/2, \beta), \beta)$-utility in the same setting.
%\jelani{Isn't this backwards? We want to say use the $2W$ algorithm to solve the $T$ case, not the other way around.}\pasin{I don't think it's backward. We're saying that if there is an algorithm with general utility guarantee for $T$, then we get another general algorithm where we replace $T$ in the utility guarantee with $2W$. So I think the statement is correct here.}
% \jelani{OK yes, I agree. Thanks!}
\end{lemma}

\begin{proof}
Let $\cA$ denote the algorithm for the former
%\jelani{`original algorithm' I think may be less ambiguous if we say `algorithm for the former' instead}
%\pasin{Changed as suggested.} 
 and $\eps'=\eps/2, \delta'=\delta/2$. The new algorithm works as follows:
\begin{itemize}
\item For all $j \in [\lceil T/W \rceil]$, run a separate copy of the $(\eps', \delta')$-DP algorithm $\cA$ on $S^{(j-1)W+1}, \dots, S^{(j+1)W}$.
\item When we would like to answer the query for $i, i + W - 1$, then use the $\lfloor i/T + 1\rfloor$th copy of the algorithm.
\end{itemize}
We can apply the parallel composition theorem (\Cref{thm:parallel-comp}) on the copies with $j = 1, 3, \ldots$; this implies that the combined algorithm for such $j$'s is $(\eps', \delta')$-DP. Similarly, we have that the combined algorithm for $j = 2, \dots$ is also $(\eps', \delta')$-DP. Then, applying basic composition ensures that the entire algorithm is indeed $(2\eps', 2\delta') = (\eps,\delta)$-DP.

%From the parallel composition theorem~\cite{McSherry10} and the fact that each time step appears in only two copies of the runs of $\cA$, the entire algorithm is indeed $(\eps, \delta)$-DP. 
The claimed accuracy follows immediately from definition.
%\badih{Should we state the parallel composition theorem referenced in the previous sentence in the Preliminaries section?}
\end{proof}

\subsubsection{Algorithm}

Given \Cref{lem:horizon-reduction-fixed-window}, we will focus only on designing the algorithm for the $T = 2W$ case. Here we show a reduction to 1d-range query. We remark that the algorithm below works even in the \emph{item-level DP} setting; indeed, we will also use it as a subroutine for item-level DP.

\begin{lemma} \label{lem:fixed-window-algo}
If there is an $(\eps, \delta)$-DP algorithm for 1d-range query with $(\alpha(M, \eps, \delta, \beta), \beta)$-utility, then there is an $(\eps, \delta)$-DP algorithm for fixed-window $\freq{\geq k}$ in the item-level DP and bundle setting when $T = 2W$ with $\left(2 \cdot \alpha\left(T+1, \frac{\eps}{4k},\frac{\delta}{8k}, \beta/2\right), \beta\right)$-utility for every $k \in \N$.
\end{lemma}

\begin{proof}
Let $M = T + 1$.
%Let $\cA$ be the $\left(\frac{\eps}{4k}, \frac{\delta}{8k}\right)$-DP algorithm for 1d-range query.
We will in fact create \emph{two} instances of 1d-range query. For clarity, we will call the first instance $\bx$ and the second instance $\bx'$. The instances are as follows:
\begin{enumerate}
\item %For each item $u \in [U]$, let $0 = t_0^u < t_1^u < \cdots < t_{m_u}^u = T + 1$ denote the time item $u$ is reached. 
Recall the definition of $m_u$ and $t^1_u, \dots, t^{m_u}_u$ from \Cref{sec:prelims}. \\
For all $u \in [U]$ and $\ell = 1, \dots, m_u - k$, do the following:
\begin{enumerate}
\item If $t_u^{\ell + k} - t_u^\ell > W$, add $t_u^\ell$ to $\bx$ and $\min\{t_u^{\ell + 1}, t_u^{\ell + k} - W\}$ to $\bx'$. \label{eq:add-elements-fixed-time}
\end{enumerate}
\item We then run the $\left(\frac{\eps}{4k},\frac{\delta}{8k}\right)$ 1d-range query algorithms on both instances $\bx, \bx'$ to get estimates $\tr_{y_1, y_2}$ of $r_{y_1, y_2}(\bx)$ and $\tr'_{y_1, y_2}$ of $r_{y_1, y_2}(\bx')$ for all $y_1, y_2 \in [M]$.
\item To answer $\freq{\geq k}(S^{[i, i + W - 1]})$, we output $|U| - \tr_{1,i} + \tr'_{1,i}$.
\end{enumerate}
Next, we claim that an item-level change can result in at most $2k$ changes to each of $\bx, \bx'$. To prove this, it suffices to show that any given element $u$ contributes to at most $k$ items added to each of $\bx, \bx'$. To see that the latter is true, let $\ell'$ be the smallest index for which $t_u^{\ell' + k} - t_{u}^{\ell'} > W$. Notice that this means $t_u^{\ell' + k} \geq W + 1$. Since the time horizon is $T = 2W$, this means that $t_u^{\ell+k} - t_u^\ell \leq W$ for all $\ell > \ell' + k$. In other words, the condition in the loop cannot be satisfied for $\ell \notin \{\ell', \dots, \ell' + k - 1\}$. Thus, the number of points added to each of $\bx, \bx'$ is at most $k$.

Given the above claim, we may apply group  privacy (\Cref{fact:group-dp}) to conclude that the entire algorithm is $(\eps, \delta)$-DP as desired.

%Consider an event-level change where $u$ is removed from the set $S^t$. Suppose that $\ell'$ is the largest index with $t = t^u_{\ell'}$ before the removal. Then, it is clear that the loop in the algorithm is exactly the same except for $\ell \in \{\ell' - k, \cdots, \ell'\}$. As a result, there can be at most $k + 1$ changes to each of $\bx, \bx'$. Therefore, by group differential privacy (\Cref{fact:group-dp}), we can conclude that the entire algorithm is $(\eps, \delta)$-DP as desired.

To see its correctness, for each $u \in [U]$, let $\ell^*(u, i)$ denote the last time step it is reached before $i$ (i.e., the largest $\ell$ such that $t_u^\ell < i$). Notice that
\begin{align*}
&\freq{\geq k}(S^{[i, i + W - 1]})
%= |\{u \mid S^{[i, i + W - 1]}_u \geq k\}|
= |U| - |\{u \mid S^{[i, i + W - 1]}_u < k\}|
%= |U| - |\{u \mid t_{\ell^*(u, i) + k} > i + W - 1\}|
= |U| - \sum_{u \in [U]} \ind[t_u^{\ell^*(u, i) + k} > i + W - 1].
\end{align*}
Observe that the two elements added in Step~(\ref{eq:add-elements-fixed-time}) get canceled out for all $\ell \ne \ell^*(u, i)$ in $r_{1,i}(\bx) - r_{1,i}(\bx')$. For $\ell = \ell^*(u, i)$, they are not canceled if and only if $t_u^{\ell + k} - W > i$. Thus,
\[
r_{1,i}(\bx) - r_{1,i}(\bx') = \sum_{u \in [U]} \ind[t_u^{\ell^*(u, i) + k} > i + W - 1].
\]
By combining the above two equations, we then have
$\freq{\geq k}(S^{[i, i + W - 1]}) = |U| - r_{1,i}(\bx) + r_{1,i}(\bx').$ The utility guarantee then follows from that of the 1d-range query algorithm.
\end{proof}

Combining~\Cref{lem:horizon-reduction-fixed-window},~\Cref{lem:fixed-window-algo}, and \Cref{thm:1d-pure} yields the following concrete bound.
\begin{corollary} \label{cor:fixed-window-alg-event-level}
For any $k \in \N$ and $\eps > 0$, there is an $\eps$-DP algorithm for fixed-window $\freq{\geq k}$ in the event-level DP and bundle setting with $(O(k \cdot \log^{1.5} W \log(1/\beta)/\eps, \beta)$-utility.
\end{corollary}

\subsubsection{Lower Bounds}

We prove two lower bounds for the problem, one based on the 1d-range query and the other based on linear queries. The latter shows that a polynomial dependency on $k$ is necessary.

\paragraph{Range Query-Based Lower Bounds} We start with the former, which is just a reduction back to the cumulative case.

\begin{lemma} \label{lem:fixed-window-lb-generic}
Let $k$ be any positive integer.
If there exists an $(\eps, \delta)$-DP algorithm for fixed-window $\freq{\geq k}$ in the event-level DP and bundle (resp. singleton) setting with $(\alpha(W, \beta), \beta)$-utility, then there exists an $(\eps, \delta)$-DP algorithm for cumulative $\freq{\geq k}$ in the event-level DP and bundle (resp. singleton) setting with $(\alpha(T, \beta), \beta)$-utility
\end{lemma}

\begin{proof}
Given an input $S^1, \dots, S^T$ for cumulative $\freq{\geq k}$, we can create an instance $S'^1, \dots, S'^{2T}$ for fixed-window $\freq{\geq k}$ by letting $W = T, S'^1 = \cdots = S'^{T} = \emptyset$ and $S'^{T + t} = S^t$ for all $t \in [T]$. It is simple to see that $\freq{\geq k}(S^{[1:t]}) = \freq{\geq k}(S'^{[t+1:t+W]})$. Therefore, running the fixed-window algorithm on the new instance allows us to answer the old instance with the same utility.
\end{proof}

Plugging the above into \Cref{lem:cum-lb-batch} and \Cref{lem:cum-lb-singleton}, we arrive at the following:

\begin{lemma} \label{lem:fixed-window-lb-batch}
Let $k$ be any positive integer and $\eps, \delta > 0$.
If there exists an $(\eps, \delta)$-DP algorithm for fixed-window $\freq{\geq k}$ in the event-level DP and bundle setting with $(\alpha(W, \beta), \beta)$-utility, then there exists an $(\eps, \delta)$-DP algorithm for 1d-range query with $(2 \cdot \alpha(M, \beta/2), \beta)$-utility.
\end{lemma}

\begin{lemma} \label{lem:fixed-window-lb-singleton}
Let $k$ be any positive integer and $\eps, \delta > 0$.
If there exists an $(\eps, \delta)$-DP algorithm for fixed-window $\freq{\geq k}$ in the event-level DP and singleton setting with $(\alpha(W, \beta), \beta)$-utility, then there exists an $(\eps, \delta)$-DP algorithm for 1d-range query on at most $n$ input points with $(2 \cdot \alpha((M + k - 1) n, \beta/2), \beta)$-utility.
\end{lemma}

\paragraph{Linear Query-Based Lower Bounds}
Next, we proceed to prove that a polynomial dependence on $k$ is necessary, by a reduction from linear queries.

\begin{lemma} \label{lem:fixed-window-lb-polyk}
Let $W, k$ be any positive integer and let $\bA$ be any $(d \times m)$ matrix such that $\min\{W , k\} \geq d + 2$. Then, if there exists an $(\eps, \delta)$-DP algorithm for fixed-window $\freq{\geq k}$ in the event-level DP and bundle setting with $\left(\alpha, \beta\right)$-utility, then there exists an $(\eps, \delta)$-DP algorithm for $\bA$-linear query with the same utility.
\end{lemma}

\begin{proof}
Assume w.l.o.g. that $k = W = d + 2$ and $T = 2W$ and let $U = m$. Given an input $\bx \in \{0,1\}^m$ to the $\bA$-linear query problem, we construct an input to the fixed-window $\freq{\geq k}$ algorithm as follows:
\begin{itemize}
\item For all $u \in [m]$ and $i \in [d]$, let $S^i_u = S^{i + W - 1}_u = A_{i, u}$.
\item Furthermore, for all $u \in [m]$, let $S^{W-1}_u = x_u + k - 2 - \sum_{j \in [d]} A_{j, u}$. \\
(For all remaining pairs $(u, t) \in [m] \times [T]$ not mentioned above, $S^t_u = 0$.)
\end{itemize}
Notice that, for all $i \in [d]$ and $u \in [m]$, we have
\begin{align*}
S^{[i, i+W-1]}_{u} 
&= \sum_{i' = i}^{i + W - 1} S^{i'}_{u} \enspace = \enspace \left(\sum_{i' = i}^{d} S^{i'}_{u}\right) + S^{W-1}_{u} + \left(\sum_{i' = W}^{i + W - 2} S^{i'}_{u}\right) + S^{i + W - 1}_u \\
&= \left(\sum_{i' = i}^{d} A_{i', u}\right) + \left(x_u + k - 2 - \sum_{j \in [d]} A_{j, u}\right) + \left(\sum_{i' = W}^{i + W - 2} A_{i' - W + 1, u}\right) + A_{i, u} \\
&= A_{i, u} + x_u + k - 2,
\end{align*}
which is equal to $k$ if and only if $x_u = 1$ and $A_{i, u} = 1$, and is less than $k$ otherwise. Hence,
\begin{align*}
\freq{\geq k}(S^{[i, i+W-1]}) = (\bA\bx)_i.
\end{align*}
Therefore, we may run the $(\eps, \delta)$-DP algorithm for $\freq{\geq k}$, and compute $\bA\bx$ with the same utility guarantee.
\end{proof}

Plugging the above into the known lower bound for linear queries (\Cref{thm:linear-query-lb} with $d, m = \Theta(\min\{k, W\})$), we get a concrete lower bound in terms of $k$:
\begin{corollary} \label{cor:fixed-window-lb-polyk-concrete}
Let $W, k$ be any positive integer and $\eps, \delta > 0$ be any sufficiently small constant. Then, there is no $(\eps, \delta)$-DP algorithm for fixed-window $\freq{\geq k}$ in the event-level DP and bundle setting with $(o(\sqrt{\min\{k, W\}}), 0.01)$-utility.
\end{corollary}

The above reduction also works in the singleton setting where we have to ``spread out'' the changes so that each time step involves only a single element.

\begin{lemma} \label{lem:fixed-window-lb-polyk-singleton}
Let $W, k$ be any positive integer and let $\bA$ be any $(d \times m)$ matrix such that $k \geq d + 2$ and $W \geq 4dm + 1$. Then, if there exists an $(\eps, \delta)$-DP algorithm for fixed-window $\freq{\geq k}$ in the event-level DP and singleton setting with $\left(\alpha, \beta\right)$-utility, then there exists an $(\eps, \delta)$-DP algorithm for $\bA$-linear query with the same utility.
\end{lemma}

\begin{proof}
Assume w.l.o.g. that $k = d + 2, W = 4dm + 1$ and $T = 2W$. Let $U = m$. Given an input $\bx \in \{0,1\}^m$ in the $\bA$-linear query problem, we construct an input to the cumulative $\freq{\geq k}$ algorithm as follows:
\begin{itemize}
\item For all $u \in [m]$ and $i \in [d]$, let $S^{mi + u}_u = S^{m(i - 1) + u + W - 1}_u = A_{i, u}$.
\item Furthermore, for all $u \in [m]$ and $\ell \in [d + 1]$, let $S^{W-1-m\ell + u}_u = \ind[\ell \leq x_u + k - 2 - \sum_{j \in [d]} A_{j, u}$]. \\
(For all remaining $t \in [T]$ not mentioned above, $S^t_u = 0$.)
\end{itemize}
It is not hard to verify that this is indeed a singleton instance.
Furthermore, notice that, for all $i \in [d]$ and $u \in [m]$, we have
\begin{align*}
S^{[mi, mi+W-1]}_{u}
&= S^{[mi, mi+4md]}_{u} 
\enspace = \enspace \sum_{i' = i}^{i + 4d - 1} S^{m i' + u}_{u} \\
&= \left(\sum_{i' = i}^{d} S^{mi' + u}_{u}\right) + \left(\sum_{i' = d+1}^{4d} S^{mi' + u}_{u}\right) + \left(\sum_{i' = 4d}^{i + 4d - 2} S^{m i' + u}_{u}\right) + S^{m (i - 1) + u + W - 1}_u \\
&= \left(\sum_{i' = i}^{d} A_{i', u}\right) + \left(x_u + k - 2 - \sum_{j \in [d]} A_{j, u}\right) + \left(\sum_{i' = 4d}^{i + 4d - 2} A_{i' - 4d + 1, u}\right) + A_{i, u} \\
&= A_{i, u} + x_u + k - 2,
\end{align*}
which is equal to $k$ if and only if $x_u = 1$ and $A_{i, u} = 1$, and is less than $k$ otherwise. Thus,
\begin{align*}
\freq{\geq k}(S^{[mi, mi+W-1]}) = (\bA\bx)_i.
\end{align*}
Therefore, we may run the $(\eps, \delta)$-DP algorithm for $\freq{\geq k}$, and compute $\bA\bx$ with the same utility.
\end{proof}

Again, plugging this into the known lower bound for linear queries (\Cref{thm:linear-query-lb} with $d, m = \Theta(\min\{k, W/k\})$), we get a concrete lower bound in terms of $k$:
\begin{corollary} \label{cor:fixed-window-lb-polyk-singleton-concrete}
Let $W, k$ be any positive integer such that $W \geq k$ and $\eps, \delta > 0$ be any sufficiently small constant. Then, there is no $(\eps, \delta)$-DP algorithm for fixed-window $\freq{\geq k}$ in the event-level DP and singleton setting with $(o(\sqrt{\min\{k, W/k\}}), 0.01)$-utility.
\end{corollary}

\subsection{Item-Level DP}

\subsubsection{Algorithms}

We start with a time-horizon reduction similar to the event-level DP case. However, for item-level DP, we cannot apply the parallel composition theorem (because item-level change can affect all the $O(T/W)$ subinstances) and instead have to apply basic/advanced composition, resulting in a reduction of $(T/W)^{O(1)}$ privacy budget to each subinstance.

\begin{lemma} 
Let $k, T, W \in \N$ be any positive integers.
Suppose that there exists an $(\eps, \delta)$-DP algorithm for fixed-window $\freq{\geq k}$ in the item-level DP and bundle setting with time horizon $2W$ with $(\alpha(\eps, \delta, \beta), \beta)$-utility. Then,
\begin{itemize}
\item There exists an $(\eps, \delta)$-DP algorithm for fixed-window $\freq{\geq k}$ in the item-level DP and bundle setting with time horizon $T$ with $\left(\alpha(\eps/(2T/W), \delta/(2T/W), \beta), \beta\right)$-utility.
\item There exists an $(\eps, \delta)$-DP algorithm for fixed-window $\freq{\geq k}$ in the item-level DP and bundle setting with time horizon $T$ with $\left(\alpha\left(\frac{\eps}{2\sqrt{4T/W \cdot \ln(2/\delta)}}, \frac{\delta}{4T/W}, \beta\right), \beta\right)$-utility.
\end{itemize}
\end{lemma}

\begin{proof}
We use exactly the same algorithm as in the proof of \Cref{lem:horizon-reduction-fixed-window}, except that this time we use basic or advanced composition theorems over the $\lceil T/W \rceil \leq 2T/W$ runs of the $(\eps', \delta')$-DP algorithm. This ensures that the entire algorithm is $(\eps, \delta)$-DP as long as we pick $\eps' = \eps/(2T/W), \delta' = \delta/(2T/W)$ or $\eps' = \frac{\eps}{2\sqrt{4T/W \cdot \ln(2/\delta)}}, \delta' = \frac{\delta}{4T/W}$.
\end{proof}

Combining the above with \Cref{lem:fixed-window-algo}, we arrive at the following:

\begin{lemma} 
Let $k, T, W \in \N$ be any positive integers.
Suppose that there exists an $(\eps, \delta)$-DP algorithm for 1d-range query with $(\alpha(M, \eps, \delta, \beta), \beta)$-utility.  Then, there exists an
\begin{itemize}
\item $(\eps, \delta)$-DP algorithm for fixed-window $\freq{\geq k}$ in the item-level DP and bundle setting with time horizon $T$ with $\left(2W+1,2\cdot\alpha(2W+1, \eps/(8kT/W), \delta/(16kT/W), \beta/2), \beta\right)$-utility.
\item $(\eps, \delta)$-DP algorithm for fixed-window $\freq{\geq k}$ in the item-level DP and bundle setting with time horizon $T$ with $\left(2\cdot \alpha\left(2W+1,\frac{\eps}{8k\sqrt{4T/W \cdot \ln(2/\delta)}}, \frac{\delta}{16kT/W}, \beta/2\right), \beta\right)$-utility.
\end{itemize}
\end{lemma}

Finally, plugging in the algorithmic bound for 1d-range query (\Cref{thm:1d-pure}) yields:

\begin{corollary} \label{cor:fixed-window-alg-item-pure}
For any $T, W, k \in \N$, and $\eps > 0$, there is an $\eps$-DP algorithm for fixed-window $\freq{\geq k}$ in the item-level DP and bundle setting with $(O(k \cdot (T/W) \cdot \log^{1.5} W \log(1/\beta)/\eps, \beta)$-utility.
\end{corollary}

\begin{corollary} \label{cor:fixed-window-alg-item-apx}
For any $k \in \N$ and $\eps, \delta > 0$, there is an $(\eps, \delta)$-DP algorithm for fixed-window $\freq{\geq k}$ in the item-level DP and bundle setting with $(O(k \cdot \sqrt{T/W \cdot \log(1/\delta)}  \log^{1.5} W \log(1/\beta)/\eps, \beta)$-utility.
\end{corollary}

\subsubsection{Lower Bounds}

First, we note that the lower bounds based on $k$ from the event-level DP case (\Cref{cor:fixed-window-lb-polyk-concrete} and \Cref{cor:fixed-window-lb-polyk-singleton-concrete}) immediately translate to this case.

Next, we prove $T^{\Omega(1)}$ lower bounds based on reductions from the 1-way marginal problem. We again start with the simpler bundle setting.

\begin{lemma} \label{lem:lb-fixed-window-event-batch}
Let $T, W, d$ be positive integers such that $T \geq W \cdot d$.
Then, if there exists an $(\eps, \delta)$-DP algorithm for fixed-window $\freq{\geq k}$ in the item-level DP and bundle setting with fixed window $W$ and time horizon $T$ with $(\alpha, \beta)$-utility, then there exists an $(\eps, \delta)$-DP algorithm for the 1-way marginal problem with the same utility.
\end{lemma}

\begin{proof}
Let $U = n$, Given an instance $x^1, \dots, x^n \in \{0, 1\}^d$ of the 1-way marginal problem, we construct the instance of fixed-window $\freq{\geq k}$ as follows: let
\begin{align*}
S^{i}_{j} &= 
\begin{cases}
k \cdot x^j_{(i-1)/W + 1} & \text{ if } W \mid (i - 1) \text{ and } i \leq W \cdot d, \\
0 & \text{ otherwise.}
\end{cases}
\end{align*}
for all $i \in [T], j \in [n]$.

Clearly, changing a single $x^j$ leads to a change in only a single item. Furthermore, it is not hard to verify that $\bbx_\ell = \freq{\geq k}(S^{[W(\ell-1)+1, W\ell]})$ for all $\ell \in [d]$. Thus, by running the $(\eps, \delta)$-DP algorithm for fixed-window $\freq{\geq k}$ on the above instance, we also get an estimate for the $1$-way marginal with the same utility.
\end{proof}

Plugging \Cref{lem:lb-fixed-window-event-batch} into the known lower bounds for the 1-way marginal problem (\Cref{thm:marginal-pure-lb} and \Cref{thm:marginal-apx-lb}) yields:
\begin{corollary}
Let $T \geq W$ be any positive integers and $\eps > 0$. Then, there is no $\eps$-DP algorithm for time-window $\freq{\geq k}$ in the item-level DP and bundle setting with $\left(o\left((T/W)/\eps\right), 0.01\right)$-utility.
\end{corollary}

\begin{corollary}
Let $T \geq W$ be any positive integers and $\eps, \delta > 0$ be such that $\delta \in (2^{-\Omega(T)}, 1/T^{1+\Omega(1)})$. Then, there is no $(\eps, \delta)$-DP algorithm for time-window $\freq{\geq k}$ in the item-level DP and bundle setting with $(o(\sqrt{(T/W)\log(1/\delta)}/\eps), 0.01)$-utility.
\end{corollary}

In the singleton case, we can use almost the same reduction as before, except that we now have to ``spread out'' the contribution across time step. This also means that we require an additional condition that the window is sufficiently large (i.e., $W \geq kn$) in the lemma below.

\begin{lemma} \label{lem:lb-fixed-window-item-singleton}
Let $T, W, d, k, n$ be positive integers such that $T \geq W \cdot d$ and $W \geq k \cdot n$.
Then, if there exists an $(\eps, \delta)$-DP algorithm for fixed-window $\freq{\geq k}$ in the item-level DP and singleton setting with fixed window $W$ and time horizon $T$ with $(\alpha, \beta)$-utility, then there exists an $(\eps, \delta)$-DP algorithm for 1-way marginal with the same utility.
\end{lemma}

\begin{proof}
Let $U = n$.  Given an instance $x^1, \dots, x^n \in \{0, 1\}^d$ of the 1-way marginal problem, we construct the instance of fixed-window $\freq{\geq k}$ as follows: let
\begin{align*}
S^{i}_{j} &= 
\begin{cases}
x^j_{\lfloor (i - 1) / W\rfloor + 1} & \text{ if } (i - k(j-1)) \mod W < k \text{ and } i \leq W \cdot d, \\
0 & \text{ otherwise,}
\end{cases}
\end{align*}
for all $i \in [T], j \in [n]$.

It is not hard to verify that $\|S^i\|_1 \leq 1$ for all $i \in [T]$, i.e., that this is a valid instance for the singleton case. Furthermore, changing a single $x^j$ leads to a change in only a single item and $\bbx_\ell = \freq{\geq k}(S^{[W(\ell-1)+1, W\ell]})$ for all $\ell \in [d]$. Thus, by running the $(\eps, \delta)$-DP algorithm for fixed-window $\freq{\geq k}$ on the above instance, we also get an estimate for the $1$-way marginal with the same utility.
\end{proof}

Similarly, plugging \Cref{lem:lb-fixed-window-item-singleton} to the known lower bounds for the 1-way marginal problem (\Cref{thm:marginal-pure-lb}  and \Cref{thm:marginal-apx-lb} with $n = \Theta(W/k), d = \Theta(T/W)$) yields:
\begin{corollary} \label{cor:lb-fixed-window-item-pure}
Let $T \geq W \geq k$ be any positive integers and $\eps > 0$. Then, there is no $\eps$-DP algorithm for fixed-window $\freq{\geq k}$ in the item-level DP and singleton setting with $\left(o\left(\min\{W/k, (T/W)/\eps \}\right), 0.01\right)$-utility.
\end{corollary}

\begin{corollary} \label{cor:lb-fixed-window-item-apx}
Let $T \geq W \geq k$ be any positive integers and $\eps, \delta > 0$ be such that $\delta \in (2^{-\Omega(W/k)}, 1/(W/k)^{1+\Omega(1)})$. Then, there is no $(\eps, \delta)$-DP algorithm for fixed-window $\freq{\geq k}$ in the item-level DP and bundle setting with $\left(o\left(\min\left\{W/k,\frac{\sqrt{(T/W) \log(1/\delta)}}{\eps}\right\}\right), 0.01\right)$-utility.
\end{corollary}

Again, we remark that achieving $W/k$ error is trivial because in the singleton case, the answer to any fixed-window query is at most $W/k$. Therefore, simply answering zero to all queries result in error at most $W/k$.

\section{Time-Window Queries}

Finally, we will present our algorithms and lower bounds for time-window queries. Again, this section is divided into two subsections based on the DP notions.

\subsection{Event-Level DP}

\subsubsection{Algorithm}

In the time-window setting, we reduce the problem to 2d-range query:

\begin{lemma} \label{lem:time-window-algo}
If there exists an $(\eps, \delta)$-DP algorithm for 2d-range query with $\left(\alpha\left(M, \eps, \delta, \beta\right), \beta\right)$-utility, then there exists an $(\eps, \delta)$-DP algorithm for time-window $\freq{\geq k}$ in the event-level DP and bundle setting with $\left(2 \cdot \alpha\left(T+1, \frac{\eps}{2(2k+1)},\frac{\delta}{4\eps(2k+1)}, \beta/2\right), \beta\right)$-utility for every $k \in \N$.
\end{lemma}

\begin{proof}
Let $M = T + 1$.
We will create \emph{two} instances of 2d-range query. For clarity, we will call the first instance $\bx$ and the second instance $\bx'$. The instances are created as follows:
\begin{enumerate}
\item Recall the definition of $m_u$ and $t^1_u, \dots, t^{m_u}_u$ from \Cref{sec:prelims}. \\
For all $u \in [U]$ and $\ell = 1, \dots, m_u - k$, do the following:
\begin{enumerate}
\item Add $(t_u^\ell + 1, t_u^{\ell + k} - 1)$ to $\bx$.
\item Add $(t_u^{\ell + 1} + 1, t_u^{\ell + k} - 1)$ to $\bx'$.
\end{enumerate}
\item We then run the $\left(\frac{\eps}{2(2k+1)},\frac{\delta}{4\eps(2k+1)}\right)$ 2d-range query algorithms on both instances $\bx, \bx'$ to get estimates $\tr_{y_1, y_2}$ of $r_{y_1, y_2}(\bx)$ and $\tr'_{y_1, y_2}$ of $r_{y_1, y_2}(\bx')$ for all $y_1, y_2 \in [M]^2$.
\item To answer $\freq{\geq k}(S^{[i, j]})$, we output $|U| - \tr_{(1, j), (i, M)} + \tr'_{(1, j), (i, M)}$.
\end{enumerate}

To prove DP, consider any two neighboring datasets $\ds = (S^1, \dots, S^T)$ and $\overline{\ds}$ where $\overline{\ds} \in (\bS^1, \dots, \bS^t)$ results from removing $u$ from $S^t$ for some $u \in [U], t \in [T]$. To avoid confusion, when we refer to $m_u, t^1_u, \dots, t^{m_u}_u$, we will be explicit about whether they are corresponding to $\ds$ or $\ods$. Let $\ell'$ denote the largest index with $t = t^{\ell'}_u(\ds)$. Note that we have $m_u(\ods) = m_n(\ds) - 1$ and
\begin{align*}
t^{\ell}_u(\ods) =
\begin{cases}
t^\ell_u(\ds) & \text{ if } \ell < \ell', \\
t^{\ell + 1}_u(\ds) & \text{ if } \ell \geq \ell',
\end{cases}
\end{align*}
for all $\ell \in [m_u(\ods)]$.

This implies the following:
\begin{itemize}
\item For all $\ell \leq \ell'$ the $\ell$th loop of the first step remains exactly the same for both $\ds$ and $\ods$.
\item For all $\ell \geq \ell' + k$ the $\ell$th loop of the first step for $\ods$ is exactly the same as the $(\ell + 1)$th loop for $\ds$.
\end{itemize}
As a result, there can be at most $2k + 1$ changes to each of $\bx, \bx'$ between the two datasets $\ds, \ods$. Thus, group privacy (\Cref{fact:group-dp}), implies that the algorithm is $(\eps, \delta)$-DP as desired.

%Let $\bm_u$ be the value of $m_u$ in $\bS$ and $\bt^1_{u}, \dots, \bt^{\bm_n}_u$ be the values of $t^1_u, \dots, t^{m_u}_u$.

%Consider an event-level change where $u$ is removed from the set $S^t$. Suppose that $\ell'$ is the largest index with $t = t^u_{\ell'}$ before the removal. Then, it is clear that the loop in the algorithm is exactly the same except for $\ell \in \{\ell' - k, \cdots, \ell'\}$. As a result, there can be at most $k + 1$ changes to each of $\bx, \bx'$. Thus, group privacy (\Cref{fact:group-dp}), implies that the algorithm is $(\eps, \delta)$-DP.

To see its correctness, for each $u \in [U]$, let $\ell^*(u, i)$ denote the last time step it is reached before $i$ (i.e., largest $\ell$ such that $t_u^\ell < i$). Notice that
\begin{align*}
\freq{\geq  k}(S^{[i, j]})
% = |\{u \mid S^{[i, j]}_u \geq k\}|
= |U| - |\{u \mid S^{[i, j]}_u < k\}|
%&= |U| - |\{u \mid t_{\ell^*(u, i) + k} > j\}| \\
= |U| - \sum_{u \in [U]} \ind[t^{\ell^*(u, i) + k}_u > j].
\end{align*}
Now, observe that the elements added gets canceled out for all $\ell \ne \ell^*(u, i)$ in $r_{(1, j), (i, M)}(\bx) - r_{(1, j), (i, M)}(\bx')$. For $\ell = \ell^*(u, i)$, they are not canceled iff $t_{\ell + k} > j$. To summarize, we have
\begin{align*}
r_{(1, j), (i, M)}(\bx) - r_{(1, j), (i, M)}(\bx') = \sum_{u \in [U]} \ind[t^{\ell^*(u, i) + k}_u > j].
\end{align*}
By combining the above two equations, we then have $\freq{\geq k}(S^{[i, j]}) = |U| - r_{(1, j), (i, M)}(\bx) + r_{(1, j), (i, M)}(\bx').$ The error guarantee follows from that of the 2d-range query algorithm.
\end{proof}

Plugging the above into \Cref{thm:2d-pure,thm:2d-apx}, we arrive at the following bounds.

\begin{corollary} \label{cor:time-window-alg-pure-concrete}
For any $k \in \N, \eps > 0$, there is an $\eps$-DP algorithm for time-window $\freq{\geq k}$ in the event-level DP and bundle setting with $(O(k \cdot \log^3 T \log(1/\beta)/\eps), \beta)$-utility.
\end{corollary}

\begin{corollary} \label{cor:time-window-alg-apx-concrete}
For any $k \in \N, \eps, \delta \in (0, 1)$, there is an $(\eps, \delta)$-DP algorithm for time-window $\freq{\geq k}$ in the event-level DP and bundle setting with $(O(k \cdot \log^2 T \sqrt{\log(k/(\eps \delta))\log(1/\beta)}/\eps), \beta)$-utility.
\end{corollary}

\subsubsection{Lower Bounds}

Once again, there are two lower bounds here, one based on 2d-range query (which is polylogarithmic in $T$) and one based on the 1-way marginal problem. In fact, the latter follows immediately from the lower bound in the fixed-window case since time-window includes fixed-window with $W = 1$ and $W = T/2$ as special cases. Specifically, the lower bounds from \Cref{cor:fixed-window-lb-polyk-concrete} and \Cref{cor:fixed-window-lb-polyk-singleton-concrete} yields the following:

\begin{corollary} \label{cor:time-window-lb-polyk-concrete}
Let $T, k$ be any positive integer and $\eps, \delta > 0$ be any sufficiently small constant. Then, there is no $(\eps, \delta)$-DP algorithm for time-window $\freq{\geq k}$ in the event-level DP and bundle setting with $(o(\sqrt{\min\{k, T\}}), 0.01)$-utility.
\end{corollary}

\begin{corollary} \label{cor:time-window-lb-polyk-singleton-concrete}
Let $T, k$ be any positive integer such that $T \geq 2k$ and $\eps, \delta > 0$ be any sufficiently small constant. Then, there is no $(\eps, \delta)$-DP algorithm for time-window $\freq{\geq k}$ in the event-level DP and singleton setting with $(o(\sqrt{\min\{k, T/k\}}), 0.01)$-utility.
\end{corollary}

The lower bound for 2d-range query works by ``reverse engineering'' the reduction above. We start by describing this for the case of bundle.

\begin{lemma} \label{lem:time-window-lb-batch}
Let $k$ be any positive integer.
If there exists an $(\eps, \delta)$-DP algorithm for time-window $\freq{\geq k}$ in the event-level DP and bundle setting with $(\alpha(T, \eps, \delta, \beta), \beta)$-utility, then there exists an $(\eps, \delta)$-DP algorithm for 2d-range query with $(4 \cdot \alpha(2M+1, \eps/2, \delta/4, \beta/4), \beta)$-utility.
\end{lemma}

\begin{proof}
Let $T = 2M + 1$ and let $U$ be sufficiently large (i.e., larger than the dataset size of the 1d-range query). Given an input $x^1, \dots, x^n \in [M]^2$ in the 2d-range query problem, we construct an input to the time-window $(\eps/2,\delta/4)$-DP $\freq{\geq k}$ algorithm as follows:
\begin{itemize}
\item For all $u \in [U]$, let $S^{M+1}_u = k - 1$ and $S^i_u = 0$ for all $i \in [T] \setminus \{M+1\}$.
\item For all $j \in [n]$, increment each of $S^{M + 1 - x^j_1}_{j}$ and $S^{M + 1 + x^j_2}_{j}$ by one.
\end{itemize}
It is not hard to verify that, for all $y \in [M]^2$, $\freq{k}(S^{[M + 1 - y_1, M + 1 + y_2]}) = n -  r_{(1,1), y}(\bx)$. Therefore, we may answer any range query $r_{y^1, y^2}(\bx)$ by outputting an estimate to 
\begin{align*}
&- \freq{\geq k}(S^{[M + 1 - y^2_1, M + 1 + y^2_2]}) - \freq{\geq k}(S^{[M + 2 - y^1_1, M + y^1_2]}) \\ 
&+ \freq{\geq k}(S^{[M + 1 - y^2_1, M + y^1_2]}) + \freq{\geq k}(S^{[M + 2 - y^1_1, M + 1 + y^2_2]}),
\end{align*}
where the estimate of each term is computed via the time-window $(\eps/2,\delta/4)$-DP $\freq{\geq k}$ algorithm. 

By group DP (\Cref{fact:group-dp}), this is indeed an $(\eps, \delta)$-DP algorithm for 2d-range query. Its utility claim follows trivially by definition.% with $(4 \cdot \alpha(2M+1, \eps/2, \delta/2, \beta/4), \beta)$-utility.
\end{proof}

Plugging the above into~\Cref{thm:2d-range-lb} gives the following lower bound in terms of $T$:

\begin{corollary} \label{cor:time-window-lb-batch-concrete}
For any sufficiently small constants $\eps, \delta, \beta > 0$, there is no $(\eps, \delta)$-DP algorithm for time-window $\freq{\geq k}$ in the event-level DP and bundle setting with $(o(\log T), \beta)$-utility
\end{corollary}

Again, our standard ``spreading out'' technique also works with the above reduction with a usual blow up on the value of $T$.

\begin{lemma} \label{lem:time-window-lb-singleton}
Let $T, k, n, M$ be any positive integers such that $T \geq 2nM + k n$.
If there exists an $(\eps, \delta)$-DP algorithm for time-window $\freq{\geq k}$ in the event-level DP and singleton setting with $(\alpha(T, \eps, \delta, \beta), \beta)$-utility, then there exists an $(\eps, \delta)$-DP algorithm for 2d-range query with $(4 \cdot \alpha(2M+1, \eps/2, \delta/4, \beta/4), \beta)$-utility.
\end{lemma}

\begin{proof}
Let $U = n$. For convenience, also let $Q = n M + 1, R = Q + (k - 1)n$. Given an input $x^1, \dots, x^n \in [M]^2$ to the 2d-range query problem, we construct an input to the time-window $(\eps/2,\delta/4)$-DP $\freq{\geq k}$ algorithm as follows:
\begin{itemize}
\item For all $u \in [U]$, let $S^{Q+(k-1)(u-1)+1}_u = \cdots = S^{Q+(k-1)u}_u = k - 1$ and $S^i_u = 0$ for all $i \in [T] \setminus \{Q+(k-1)(u-1)+1, \dots, Q+(k-1)u\}$.
\item For all $j \in [n]$, increment each of $S^{Q - n x^j_1 + u}_{j}$ and $S^{R + n x^j_2 + u}_{j}$ by one.
\end{itemize}
It is not hard to verify that, for all $y \in [M]^2$, $\freq{k}(S^{[Q - n y_1, R + n(y_2 + 1) - 1]}) = n -  r_{(1,1), y}(\bx)$. Therefore, we may answer any range query $r_{y^1, y^2}(\bx)$ by outputting an estimate to 
\begin{align*}
&- \freq{\geq k}(S^{[Q - n y^2_1, R + n(y^2_2 + 1) - 1]}) - \freq{\geq k}(S^{[Q - n(y^1_1 - 1), R + n y^1_2 - 1]}) \\ 
&+ \freq{\geq k}(S^{[Q - n y^2_1, R + n y^1_2 - 1]}) + \freq{\geq k}(S^{[Q - n(y^1_1 - 1), R + n (y^2_2 + 1) - 1]}),
\end{align*}
where the estimate of each term is computed via the time-window $(\eps/2,\delta/4)$-DP $\freq{\geq k}$ algorithm. 

By group DP (\Cref{fact:group-dp}), this is indeed an $(\eps, \delta)$-DP algorithm for 2d-range query. Its utility claim follows trivially by definition.% with $(4 \cdot \alpha(2M+1, \eps/2, \delta/2, \beta/4), \beta)$-utility.
\end{proof}

Plugging the above into~\Cref{thm:2d-range-lb} with $M = \sqrt[3]{T/k}, n = \Theta(M^2)$ gives the following lower bound in terms of $T/k$:

\begin{corollary} \label{cor:time-window-lb-singleton-concrete}
For any sufficiently small constants $\eps, \delta, \beta > 0$, there is no $(\eps, \delta)$-DP algorithm for time-window $\freq{\geq k}$ in the event-level DP and singleton setting with $(o(\log(T/k)), \beta)$-utility
\end{corollary}

\subsection{Item-Level DP}

\subsubsection{Algorithms}

The algorithm for Item-Level DP is based on a rather simple approach of running the cumulative $\freq{\geq k}$ algorithm with different starting times, and then applying the composition theorems to account for the privacy budget. This is formalized below.

\begin{lemma}
Let $k, T \in \N$ be any positive integer.
Suppose that there exists an $(\eps, \delta)$-DP algorithm for cumulative $\freq{\geq k}$ in the item-level DP and bundle setting with time horizon $T$ with $(\alpha(\eps, \delta, \beta), \beta)$-utility. Then,
\begin{itemize}
\item There exists an $(\eps, \delta)$-DP algorithm for time-window $\freq{\geq k}$ in the item-level DP and bundle setting with time horizon $T$ with $\left(\alpha(\eps/T, \delta/T, \beta), \beta\right)$-utility.
\item There exists an $(\eps, \delta)$-DP algorithm for time-window $\freq{\geq k}$ in the item-level DP and bundle setting with time horizon $T$ with $\left(\alpha\left(\frac{\eps}{2\sqrt{2T\ln(2/\delta)}}, \frac{\delta}{2T}, \beta\right), \beta\right)$-utility.
\end{itemize}
\end{lemma}

\begin{proof}
We simply run the $(\eps', \delta')$-DP algorithm for cumulative $\freq{\geq k}$ starting at time $1, \dots, T$. Since we run the $(\eps', \delta')$-DP algorithm a total of $T$ times on the input dataset, the basic and advanced composition theorems respectively ensure that the entire algorithm is $(\eps, \delta)$-DP as long as we pick $\eps' = \eps/T, \delta' = \delta/T$ and $\eps' = \frac{\eps}{2\sqrt{2T\ln(2/\delta)}}, \delta' = \frac{\delta}{2T}$, respectively.
\end{proof}

Combining the above lemma with our cumulative algorithm from \Cref{cor:cumulative-alg}, we immediately arrive at the following:

\begin{corollary} \label{cor:time-window-item-batch-pure}
For any $k \in \N$ and $\eps > 0$, there is an $\eps$-DP algorithm for time-window $\freq{\geq k}$ in the item-level DP and bundle setting with $(O(T \cdot \log^{1.5}T\log(1/\beta)/\eps, \beta)$-utility.
\end{corollary}

\begin{corollary} \label{cor:time-window-item-batch-apx}
For any $k \in \N$ and $\eps, \delta \in (0, 1)$, there is an $(\eps, \delta)$-DP algorithm for time-window $\freq{\geq k}$ in the item-level DP and bundle setting with $(O(\sqrt{T \log(T/\delta)} \cdot \log^{1.5}T\log(1/\beta)/\eps, \beta)$-utility.
\end{corollary}

For the singleton setting, we can get an improved bound by ``time compression'' technique of \cite{jain2021price} as formalized below.

\begin{lemma} \label{lem:generic-singleton}
Let $T'$ be any positive integer and $\eps, \delta > 0$.
If there exists an $(\eps, \delta)$-DP algorithm for $\freq{\geq k}$ with utility $(\alpha(T, \beta), \beta)$ in the bundle setting with item-level DP, then there exists an $(\eps, \delta)$-DP algorithm for $\freq{\geq k}$ with $(\alpha(\lceil T/T' \rceil, \beta) + T', \beta)$-utility in the singleton setting (both for item-level DP and event-level DP).
\end{lemma}

\begin{proof}
Let $S^1, \cdots, S^T$ denote the input to the singleton setting. We construct an input to the bundle setting by, for all $i \in [\lceil T / T'\rceil$ letting $\tS^i := S^{T'(i - 1)} + \cdots + S^{\min\{T'(i+1), T\}}$. When we want to answer the query for the time from $t_1$ to $t_2$ in the original instance, we instead use the answer for the query from time $\lceil t_1 / T'\rceil$ to $\lceil t_2 / T' \rceil$ in the new instance. From the fact that the original instance has only a single item per time step, it is not hard to see that this only introduces at most $T'$ additional error. The claimed DP guarantee is immediate.
\end{proof}

By applying \Cref{lem:generic-singleton} to the above corollaries (with $T' = \sqrt{T/\eps}$ and $T' = \sqrt[3]{T/\eps^2}$ respectively), we get:

\begin{corollary} \label{cor:time-window-item-singleton-pure}
For any $k \in \N$ and $\eps > 0$, there is an $\eps$-DP algorithm for time-window $\freq{\geq k}$ in the item-level DP and singleton setting with $(O(\sqrt{T/\eps} \cdot \log^{1.5}T\log(1/\beta), \beta)$-utility.
\end{corollary}

\begin{corollary} \label{cor:time-window-item-singleton-apx}
For any $k \in \N$ and $\eps, \delta \in (0, 1)$, there is an $(\eps, \delta)$-DP algorithm for time-window $\freq{\geq k}$ in the item-level DP and singleton setting with $(O(\sqrt[3]{T/\eps^2}\sqrt{\log(T/\delta)} \cdot \log^{1.5}T\log(1/\beta), \beta)$-utility.
\end{corollary}

\subsubsection{Lower Bounds}

Since time-window includes fixed-window with $W = 1$ and $W = kn$ as special cases, the lower bounds from~\Cref{lem:lb-fixed-window-event-batch} and~\Cref{lem:lb-fixed-window-item-singleton} immediately translate to the following:

\begin{lemma} \label{lem:lb-time-window-event-batch}
Let $T, d$ be positive integers such that $T \geq d$.
Then, if there exists an $(\eps, \delta)$-DP algorithm for time-window $\freq{\geq k}$ in the item-level DP and bundle setting with time horizon $T$ with $(\alpha, \beta)$-utility, then there exists an $(\eps, \delta)$-DP algorithm for 1-way marginal with the same utility.
\end{lemma}

\begin{lemma} \label{lem:lb-time-window-event-singleton}
Let $T, d, k, n$ be positive integers such that $T \geq knd$.
Then, if there exists an $(\eps, \delta)$-DP algorithm for time-window $\freq{\geq k}$ in the item-level DP and singleton setting with time horizon $T$ with $(\alpha, \beta)$-utility, then there exists an $(\eps, \delta)$-DP algorithm for 1-way marginal with the same utility.
\end{lemma}

Plugging \Cref{lem:lb-time-window-event-batch} into the known lower bounds for the 1-way marginal problem (\Cref{thm:marginal-pure-lb} and \Cref{thm:marginal-apx-lb}) yields:
\begin{corollary} \label{cor:time-window-item-batch-pure-lb}
Let $T$ be any positive integer and $\eps > 0$. Then, there is no $\eps$-DP algorithm for time-window $\freq{\geq k}$ in the item-level DP and bundle setting with $(o(T/\eps), 0.01)$-utility.
\end{corollary}

\begin{corollary} \label{cor:time-window-item-batch-apx-lb}
Let $T$ be any positive integer and $\eps, \delta > 0$ be such that $\delta \in (2^{-\Omega(T)}, 1/T^{1+\Omega(1)})$. Then, there is no $(\eps, \delta)$-DP algorithm for time-window $\freq{\geq k}$ in the item-level DP and bundle setting with $(o(\sqrt{T\log(1/\delta)}/\eps), 0.01)$-utility.
\end{corollary}

Similarly, plugging \Cref{lem:lb-time-window-event-singleton} to the known lower bounds for 1-way marginal (\Cref{thm:marginal-pure-lb} with $n = \Theta(T/(k\eps), d = \Theta(\eps T/k)$ and \Cref{thm:marginal-apx-lb} with $n = \Theta\left(\sqrt[3]{\frac{T \log(1/\delta)}{\eps^2 k}}\right) d = \left(\sqrt[3]{\frac{\eps^2 T^2 }{k^2 \log(1/\delta)}}\right)$) yields:
\begin{corollary} \label{cor:time-window-item-singleton-pure-lb}
Let $T, k$ be any positive integers and $\eps > 0$ such that $T \geq k / \eps$. Then, there is no $\eps$-DP algorithm for time-window $\freq{\geq k}$ in the item-level DP and singleton setting with $\left(o\left(\sqrt{\frac{T}{\eps k}}\right), 0.01\right)$-utility.
\end{corollary}

\begin{corollary} \label{cor:time-window-item-singleton-apx-lb}
Let $T$ be any positive integer and $\eps, \delta > 0$ be such that $\delta \in (2^{-\Omega(T)}, 1/T^{1+\Omega(1)})$ and $T \geq k \sqrt{\log(1/\delta)} / \eps$. Then, there is no $(\eps, \delta)$-DP algorithm for time-window $\freq{\geq k}$ in the item-level DP and bundle setting with $\left(o\left(\sqrt[3]{\frac{T \log(1/\delta)}{\eps^2 k}}\right), 0.01\right)$-utility.
\end{corollary}

All of the lower bounds are tight to within polylogarithmic factors and the dependency on $k$ compared to the algorithms given in the previous subsection. (Note that the lower bounds in terms of $k$ still carries over from the event-label DP case above, i.e., \Cref{cor:time-window-lb-polyk-concrete} and \Cref{cor:time-window-lb-polyk-singleton-concrete}.)

\section{Conclusions and Open Questions}\label{sec:conclusion}

In this work, we consider several variants of the tasks of counting distinct and $k$-occurring elements in time windows, including cumulative/time-window/fixed-window queries, user/item-level DP, and bundle/singleton settings. In each setting, we determine optimal error bounds that are tight up to polylogarithmic factors in $T$ or $W$ and the dependency on $k$. As mentioned earlier, our work is closely related to the continual release model~\cite{DworkNPR10,ChanSS11}. In fact, it is simple to modify our algorithms to work in the following setting: at time step $t$, the algorithm receives $S^t$ and must immediately answer all queries $(t_1, t_2)$ such that $t_2 = t$. This can be done by using the 1d-/2d-range query algorithms that work in the similar setting~\cite{DworkNPR10,ChanSS11,DworkNRR15}.

An interesting research direction is to close the gaps in our results. For example, can we get a bound of the form $O(k + \log^{1.5} W)$ in the fixed-window event-level DP setting or a bound of the form $O(k + \log^3 T)$ in the time-window event-level DP setting? We conjecture that answering these questions requires ``white-box'' solutions beyond directly reducing to/from range query or linear queries. %It is also interesting whether such white-box approaches can be used to improve concrete numerical bounds.

Finally, we note that in this work, we have not considered memory constraints on the algorithm. It is an interesting direction for future work to consider the trade-off between memory, privacy, and utility  (e.g., similarly to \cite{smith2020flajolet}), which is an important aspect both in theory and in practice.

\ifarxiv
\bibliographystyle{alpha}
\else
\addcontentsline{toc}{section}{References}
\bibliographystyle{plainurl}
\fi
\bibliography{ref}

\newcommand{\etalchar}[1]{$^{#1}$}
\begin{thebibliography}{MMNW11}

\bibitem[ABRS03]{akella2003detecting}
Aditya Akella, Ashwin Bharambe, Mike Reiter, and Srinivasan Seshan.
\newblock Detecting {DDoS attacks on ISP} networks.
\newblock In {\em Workshop on Management and Processing of Data Streams}, 2003.

\bibitem[AGPR99]{acharya1999aqua}
Swarup Acharya, Phillip~B Gibbons, Viswanath Poosala, and Sridhar Ramaswamy.
\newblock The {Aqua} approximate query answering system.
\newblock In {\em SIGMOD}, pages 574--576, 1999.

\bibitem[AS17]{agarwal2017price}
Naman Agarwal and Karan Singh.
\newblock The price of differential privacy for online learning.
\newblock In {\em ICML}, pages 32--40, 2017.

\bibitem[BBS18]{breitwieser2018krakenuniq}
Florian~P Breitwieser, DN~Baker, and Steven~L Salzberg.
\newblock Krakenuniq: confident and fast metagenomics classification using
  unique $k$-mer counts.
\newblock {\em Genome Biology}, 19(1):1--10, 2018.

\bibitem[BCJM21]{balcer2021connecting}
Victor Balcer, Albert Cheu, Matthew Joseph, and Jieming Mao.
\newblock Connecting robust shuffle privacy and pan-privacy.
\newblock In {\em SODA}, pages 2384--2403, 2021.

\bibitem[BFM{\etalchar{+}}13]{BolotFMNT13}
Jean Bolot, Nadia Fawaz, S.~Muthukrishnan, Aleksandar Nikolov, and Nina Taft.
\newblock Private decayed predicate sums on streams.
\newblock In {\em ICDT}, pages 284--295, 2013.

\bibitem[BL19]{baker2019dashing}
Daniel~N Baker and Ben Langmead.
\newblock Dashing: fast and accurate genomic distances with {HyperLogLog}.
\newblock {\em Genome Biology}, 20(1):1--12, 2019.

\bibitem[BUV18]{BunUV18}
Mark Bun, Jonathan~R. Ullman, and Salil~P. Vadhan.
\newblock Fingerprinting codes and the price of approximate differential
  privacy.
\newblock {\em {SIAM} J. Comput.}, 47(5):1888--1938, 2018.

\bibitem[CdGK10]{CheongGK10}
Yunjae Cheong, Federico de~Gregorio, and Kihan Kim.
\newblock The power of reach and frequency in the age of digital advertising:
  Offine and online media demand different metrics.
\newblock {\em J. Advertising Res.}, 50, 2010.

\bibitem[CGKM21]{chen2020distributed}
Lijie Chen, Badih Ghazi, Ravi Kumar, and Pasin Manurangsi.
\newblock On distributed differential privacy and counting distinct elements.
\newblock In {\em ITCS}, pages 56:1--56:18, 2021.

\bibitem[Che22]{quantcast2022}
Steven Chen.
\newblock Ara, tell me what my campaign forecast looks like for today.
\newblock
  \url{https://www.quantcast.com/blog/ara-tell-me-what-my-campaign-forecast-looks-like-for-today/},
  2022.

\bibitem[CR22]{cardoso2022differentially}
Adrian~Rivera Cardoso and Ryan Rogers.
\newblock Differentially private histograms under continual observation:
  Streaming selection into the unknown.
\newblock In {\em AISTATS}, pages 2397--2419, 2022.

\bibitem[CSS11]{ChanSS11}
T.{-}H.~Hubert Chan, Elaine Shi, and Dawn Song.
\newblock Private and continual release of statistics.
\newblock {\em {ACM} Trans. Inf. Syst. Secur.}, 14(3):26:1--26:24, 2011.

\bibitem[DKM{\etalchar{+}}06]{dwork2006our}
Cynthia Dwork, Krishnaram Kenthapadi, Frank McSherry, Ilya Mironov, and Moni
  Naor.
\newblock Our data, ourselves: Privacy via distributed noise generation.
\newblock In {\em EUROCRYPT}, pages 486--503, 2006.

\bibitem[DMNS06]{DworkMNS06}
Cynthia Dwork, Frank McSherry, Kobbi Nissim, and Adam~D. Smith.
\newblock Calibrating noise to sensitivity in private data analysis.
\newblock In {\em TCC}, pages 265--284, 2006.

\bibitem[DMT07]{DworkMT07}
Cynthia Dwork, Frank McSherry, and Kunal Talwar.
\newblock The price of privacy and the limits of {LP} decoding.
\newblock In {\em STOC}, pages 85--94, 2007.

\bibitem[DN03]{DinurN03}
Irit Dinur and Kobbi Nissim.
\newblock Revealing information while preserving privacy.
\newblock In {\em PODS}, pages 202--210, 2003.

\bibitem[DNP{\etalchar{+}}10]{DworkNPRY10}
Cynthia Dwork, Moni Naor, Toniann Pitassi, Guy~N. Rothblum, and Sergey
  Yekhanin.
\newblock Pan-private streaming algorithms.
\newblock In {\em ICS}, pages 66--80, 2010.

\bibitem[DNPR10]{DworkNPR10}
Cynthia Dwork, Moni Naor, Toniann Pitassi, and Guy~N. Rothblum.
\newblock Differential privacy under continual observation.
\newblock In {\em STOC}, pages 715--724, 2010.

\bibitem[DNRR15]{DworkNRR15}
Cynthia Dwork, Moni Naor, Omer Reingold, and Guy~N. Rothblum.
\newblock Pure differential privacy for rectangle queries via private
  partitions.
\newblock In {\em ASIACRYPT}, pages 735--751, 2015.

\bibitem[DR14]{dwork2014algorithmic}
Cynthia Dwork and Aaron Roth.
\newblock {The Algorithmic Foundations of Differential Privacy}.
\newblock {\em Found. Trends Theor. Comput. Sci.}, 9(3-4):211--407, 2014.

\bibitem[DRV10]{DworkRV10}
Cynthia Dwork, Guy~N. Rothblum, and Salil~P. Vadhan.
\newblock Boosting and differential privacy.
\newblock In {\em FOCS}, pages 51--60, 2010.

\bibitem[EVF03]{estan2003bitmap}
Cristian Estan, George Varghese, and Mike Fisk.
\newblock Bitmap algorithms for counting active flows on high speed links.
\newblock In {\em IMC}, pages 153--166, 2003.

\bibitem[FHO21]{fichtenberger2021differentially}
Hendrik Fichtenberger, Monika Henzinger, and Wolfgang Ost.
\newblock Differentially private algorithms for graphs under continual
  observation.
\newblock In {\em ESA}, pages 42:1--42:16, 2021.

\bibitem[GKK{\etalchar{+}}22]{ghazi2022multiparty}
Badih Ghazi, Ben Kreuter, Ravi Kumar, Pasin Manurangsi, Jiayu Peng, Evgeny
  Skvortsov, Yao Wang, and Craig Wright.
\newblock Multiparty reach and frequency histogram: Private, secure, and
  practical.
\newblock {\em PoPETS}, 2022(1):373--395, 2022.

\bibitem[GTS13]{guha2013nearly}
Abhradeep Guha~Thakurta and Adam Smith.
\newblock ({N}early) optimal algorithms for private online learning in
  full-information and bandit settings.
\newblock {\em NIPS}, 26, 2013.

\bibitem[HNH13]{heule2013hyperloglog}
Stefan Heule, Marc Nunkesser, and Alexander Hall.
\newblock Hyperloglog in practice: Algorithmic engineering of a state of the
  art cardinality estimation algorithm.
\newblock In {\em EDBT}, pages 683--692, 2013.

\bibitem[HT10]{HardtT10}
Moritz Hardt and Kunal Talwar.
\newblock On the geometry of differential privacy.
\newblock In {\em STOC}, pages 705--714, 2010.

\bibitem[HU22]{HenzingerU22}
Monika Henzinger and Jalaj Upadhyay.
\newblock Constant matters: Fine-grained complexity of differentially private
  continual observation using completely bounded norms.
\newblock {\em arXiv:2202.11205}, 2022.

\bibitem[JKT12]{jain2012differentially}
Prateek Jain, Pravesh Kothari, and Abhradeep Thakurta.
\newblock Differentially private online learning.
\newblock In {\em COLT}, pages 1--34, 2012.

\bibitem[JRSS22]{jain2021price}
Palak Jain, Sofya Raskhodnikova, Satchit Sivakumar, and Adam Smith.
\newblock The price of differential privacy under continual observation.
\newblock In {\em TPDP@ICML}, 2022.

\bibitem[LH98]{LeckenbyH98}
John~D. Leckenby and Jongpil Hong.
\newblock Using reach/frequency for web media planning.
\newblock {\em J. Advertising Res.}, 38, 1998.

\bibitem[McS10]{McSherry10}
Frank McSherry.
\newblock Privacy integrated queries: an extensible platform for
  privacy-preserving data analysis.
\newblock {\em Commun. {ACM}}, 53(9):89--97, 2010.

\bibitem[MMNW11]{MirMNW11}
Darakhshan~J. Mir, S.~Muthukrishnan, Aleksandar Nikolov, and Rebecca~N. Wright.
\newblock Pan-private algorithms via statistics on sketches.
\newblock In Maurizio Lenzerini and Thomas Schwentick, editors, {\em PODS},
  pages 37--48, 2011.

\bibitem[MN12]{MuthukrishnanN12}
S.~Muthukrishnan and Aleksandar Nikolov.
\newblock Optimal private halfspace counting via discrepancy.
\newblock In {\em STOC}, pages 1285--1292, 2012.

\bibitem[MRT22]{McMahanRT22}
Brendan McMahan, Keith Rush, and Abhradeep~Guha Thakurta.
\newblock Private online prefix sums via optimal matrix factorizations.
\newblock {\em arXiv:2202.08312}, 2022.

\bibitem[PAK19]{perrier2018private}
Victor Perrier, Hassan~Jameel Asghar, and Dali Kaafar.
\newblock Private continual release of real-valued data streams.
\newblock In {\em NDSS}, 2019.

\bibitem[PBM{\etalchar{+}}03]{padmanabhan2003multi}
Sriram Padmanabhan, Bishwaranjan Bhattacharjee, Tim Malkemus, Leslie Cranston,
  and Matthew Huras.
\newblock Multi-dimensional clustering: A new data layout scheme in {DB2}.
\newblock In {\em SIGMOD}, pages 637--641, 2003.

\bibitem[PGF02]{palmer2002anf}
Christopher~R Palmer, Phillip~B Gibbons, and Christos Faloutsos.
\newblock {ANF:} a fast and scalable tool for data mining in massive graphs.
\newblock In {\em KDD}, pages 81--90, 2002.

\bibitem[PHIS96]{poosala1996improved}
Viswanath Poosala, Peter~J Haas, Yannis~E Ioannidis, and Eugene~J Shekita.
\newblock Improved histograms for selectivity estimation of range predicates.
\newblock {\em SIGMOD Record}, 25(2):294--305, 1996.

\bibitem[SAC{\etalchar{+}}89]{selinger1989access}
P~Griffiths Selinger, Morton~M Astrahan, Donald~D Chamberlin, Raymond~A Lorie,
  and Thomas~G Price.
\newblock Access path selection in a relational database management system.
\newblock In {\em Readings in Artificial Intelligence and Databases}, pages
  511--522. Morgan Kaufmann, 1989.

\bibitem[SDNR96]{shukla1996storage}
Amit Shukla, Prasad Deshpande, Jeffrey~F Naughton, and Karthikeyan Ramasamy.
\newblock Storage estimation for multidimensional aggregates in the presence of
  hierarchies.
\newblock In {\em VLDB}, pages 522--531, 1996.

\bibitem[SLM{\etalchar{+}}18]{song2018differentially}
Shuang Song, Susan Little, Sanjay Mehta, Staal Vinterbo, and Kamalika
  Chaudhuri.
\newblock Differentially private continual release of graph statistics.
\newblock {\em arXiv:1809.02575}, 2018.

\bibitem[SSGT20]{smith2020flajolet}
Adam Smith, Shuang Song, and Abhradeep Guha~Thakurta.
\newblock The {Flajolet--Martin} sketch itself preserves differential privacy:
  Private counting with minimal space.
\newblock In {\em NeurIPS}, pages 19561--19572, 2020.

\bibitem[SU16]{SteinkeU16}
Thomas Steinke and Jonathan~R. Ullman.
\newblock Between pure and approximate differential privacy.
\newblock {\em J. Priv. Confidentiality}, 7(2), 2016.

\bibitem[Vad17]{vadhan2017complexity}
Salil Vadhan.
\newblock {\em The Complexity of Differential Privacy}.
\newblock Springer, 2017.

\bibitem[{Wik}21]{enwiki:1021978492}
{Wikipedia contributors}.
\newblock Effective frequency --- {Wikipedia}{,} the free encyclopedia.
\newblock
  \url{https://en.wikipedia.org/w/index.php?title=Effective_frequency&oldid=1021978492},
  2021.
\newblock [Online; accessed 18-May-2022].

\end{thebibliography}

% We point out that a multi-party computation (MPC) extension of this same setting has recently been studied in \cite{ghazi2022multiparty}.

% Following prior work \cite{BolotFMNT13}, we explore differentially private mechanisms for cumulative reach.

% \badih{Mention, perhaps in the Future Directions section, that in this paper we study the problem without memory constraints. But in theory and practice, the trade-off between memory, privacy and accuracy is important.}

\ifarxiv

\appendix

\section{Exact-$k$-Occurring Items}
\label{sec:exact-freq}

We may also define a query $\freq{=k}$ as $\freq{= k}(S) := |\{u \in [U] \mid S_u = k\}|$, i.e., the number of items that appear exactly $k$ times in $S$.

A simple observation is that $\freq{=k}$ can be easily computed from $\freq{\geq k}$ and $\freq{\geq k + 1}$. This is formalized below; note that this reduction works in all settings we considered.

\begin{lemma} \label{lem:freq-at-least-v-freq-equal}
If there exists an $(\eps, \delta)$-DP algorithm for $\freq{\geq k}$ that has $(\alpha(k, \eps, \delta, \beta), \beta)$-utility for all $k \in \N$, then there exists an $(\eps, \delta)$-DP algorithm for $\freq{=k}$ algorithm with $O(\alpha(k, \eps/2, \delta/2, \beta) + \alpha(k + 1, \eps/2, \delta/2, \beta/2), \beta)$-utility for all $k \in \N$.
\end{lemma}

\begin{proof}
Simply run the $(\eps/2, \delta/2)$-DP algorithm for $\freq{\geq k}$ and the $(\eps/2, \delta/2)$-DP algorithm for $\freq{\geq k + 1}$. By subtracting these two results, we get an $(\eps, \delta)$-DP estimate for $\freq{=k}$ with $O(\alpha(k, \eps/2, \delta/2, \beta) + \alpha(k + 1, \eps/2, \delta/2, \beta/2), \beta)$-utility.
\end{proof}

The above lemma translates all of our upper bounds for $\freq{\geq k}$ to $\freq{=k}$ as well (with a small constant overhead to $\alpha$). As for the lower bounds, it is not hard to check in our proofs that our hard instances always have $\freq{\geq k} = \freq{=k}$ and therefore they also applies for $\freq{=k}$.

\fi

\end{document}